\documentclass[onecolumn,nofootinbib,notitlepage,11pt]{revtex4-1}

\usepackage[ colorlinks = true,
             linkcolor = blue,
             urlcolor  = blue,
             citecolor = red,
             anchorcolor = green,
]{hyperref}

\usepackage{mtunicode}
\usepackage{qittools}

\let\originalleft\left
\let\originalright\right
\renewcommand{\left}{\mathopen{}\mathclose\bgroup\originalleft}
\renewcommand{\right}{\aftergroup\egroup\originalright}

\begin{document}

\title{Identifying the Information Gain of a Quantum Measurement}

\author{Mario Berta}
\email[]{berta@phys.ethz.ch}
\affiliation{Institute for Theoretical Physics, ETH Zurich, 8093 Z\"urich, Switzerland}
\author{Joseph M.~Renes}
%\email[]{renes@phys.ethz.ch}
\affiliation{Institute for Theoretical Physics, ETH Zurich, 8093 Z\"urich, Switzerland}
\author{Mark M.~Wilde}
%\email[]{...}
\affiliation{School of Computer Science, McGill University, Montr\'eal, Qu\'ebec, Canada}

\begin{abstract}
We show that quantum-to-classical channels, i.e., quantum measurements, can be asymptotically simulated by an amount of classical communication equal to the quantum mutual information of the measurement, if sufficient shared randomness is available. This result generalizes Winter's measurement compression theorem for fixed independent and identically distributed inputs [Winter, CMP 244 (157), 2004] to arbitrary inputs, and more importantly, it identifies the quantum mutual information of a measurement as the information gained by performing it, independent of the input state on which it is performed. Our result is a generalization of the classical reverse Shannon theorem to quantum-to-classical channels. In this sense, it can be seen as a quantum reverse Shannon theorem for quantum-to-classical channels, but with the entanglement assistance and quantum communication replaced by shared randomness and classical communication, respectively. The proof is based on a novel one-shot state merging protocol for ``classically coherent states'' as well as the post-selection technique for quantum channels, and it uses techniques developed for the quantum reverse Shannon theorem [Berta {\it et al}., CMP 306 (579), 2011].
\end{abstract}

\maketitle

\section{Introduction}

Measurement is an integral part of quantum theory. It is the means by which
we gather information about a quantum system. Although the classical notion of a measurement
is rather straightforward, the quantum notion of measurement has
been the subject of much thought and debate \cite{BohrEinstein1949}. One interpretation is that
the act of measurement on a quantum system
causes it to abruptly jump or ``collapse'' into one of several possible states with some probability,
an evolution seemingly different from the smooth, unitary transitions
resulting from Schr\"odinger's wave equation. Some have advocated for a measurement postulate
in quantum theory \cite{dirac82}, while others have advocated that our understanding
of quantum measurement should follow from other postulates~\cite{Z09}.

In spite of the aforementioned difficulties in understanding and interpreting quantum measurement,
there is a precise question that one can formulate concerning it:
\begin{quote}
\textit{How much information is gained by performing a given quantum measurement?}
\end{quote}
This question has a rather long history, which to our knowledge begins with the work of
Groenewold~\cite{G71}. In 1971, Groenewold argued on intuitive grounds for the following ``entropy reduction'' 
to quantify the information gained by performing a quantum measurement:
\begin{equation}
H(\rho) - \sum_x p_x H(\rho_x), \label{eq:groenewold}
\end{equation} 
where $\rho$ is the initial state before the measurement occurs, $\{ p_x , \rho_x \}$ 
is the post-measurement ensemble induced by
the measurement, and $H(\sigma) \equiv -\text{tr}[ \sigma \log \sigma]$ is the
von Neumann entropy of a state $\sigma$. The intuition behind this measure is that it
quantifies the reduction in uncertainty after performing a quantum measurement on a
quantum system in state $\rho$, and its form is certainly reminiscent of a
Holevo-like quantity \cite{Holevo73}, although the
classical data in the above Groenewold quantity appears at the
\emph{output} of the process rather than at the \emph{input} as in the case of the Holevo quantity.
Groenewold left open the question of whether
this quantity is non-negative for all measurements, and Lindblad proved that non-negativity
holds whenever the measurement is of the von Neumann-L\"uders kind (projecting onto an eigenspace of
an observable) \cite{L72}. Ozawa then settled the matter 
by proving that the above quantity is non-negative if and only if the post-measurement states are of the form
\begin{equation}
\rho_x = \frac{M_x \rho M_x^\dag}{\text{tr}[M_x^\dag M_x \rho]} \label{eq:special-measurement} ,
\end{equation}
for some operators $\{M_x\}$ such that
$\sum_x  M_x^\dag M_x = \id$ \cite{O86}. Such measurements are termed ``efficient'', and differ from general measurements as the latter may have several operators $M_{x,s}$ corresponding to the result $x$~\cite{FJ01}.
%\footnote{These post-measurement states are indeed a special case of those
%resulting from performing a more general quantum instrument.
%For such a device, the post-measurement states take the form
%$\mathcal{M}_x(\rho) / \text{tr}\{ \mathcal{M}_x(\rho) \}$ where $\mathcal{M} \equiv \{ \mathcal{M}_x \}$
%is a collection of completely positive, trace non-increasing maps such that the sum map $\sum_x \mathcal{M}_x$ is trace preserving.}

The fact that the quantity in \eqref{eq:groenewold} can become negative for some quantum measurements excludes
it from being a generally appealing measure of information gain. To remedy this situation,
Buscemi~\textit{et al.}~later advocated for the following
measure to characterize the information gain of a quantum measurement
when acting upon a particular state $\rho$ \cite{BHH08,PhysRevA.82.052103,shirokov:052202,Wilde12}:
\begin{equation}
I(X:R)_{\omega} \, , \label{eq:winter-info-gain}
\end{equation}
where $I(X:R)_{\omega} \equiv H(X)_\omega + H(R)_\omega - H(XR)_\omega$ is the quantum mutual information of
the following state:
\begin{equation}
\omega_{XR} \equiv \sum_x \vert x\rangle \langle x\vert_X \otimes
\text{tr}_A\{( \mathcal{M}_x \otimes \mathcal{I}_R  ) (\ket{\rho}\bra{\rho}_{AR}) \} .
\label{eq:info-gain-state}
\end{equation}
The register $X$ is a classical register containing the outcome of the measurement, $\mathcal{M} \equiv \{ \mathcal{M}_x \}$
is a collection of completely positive, trace non-increasing maps characterizing the measurement
(for which the sum map $\sum_x \mathcal{M}_x$ is trace preserving), $\mathcal{I}$ is the identity map, and
$\ket{\rho}_{AR}$ is a purification of the initial state $\rho$ on system $A$
to a purifying system $R$. The advantages of the measure
of information gain in \eqref{eq:winter-info-gain} are as follows:
\begin{itemize}
\item It is non-negative.
\item It reduces to Groenewold's quantity in \eqref{eq:groenewold}
for the special case of measurements of the form in \eqref{eq:special-measurement} \cite{BHH08}.
\item It characterizes the trade-off between information and disturbance in quantum measurements~\cite{BHH08}. 
\item It has an operational interpretation in Winter's measurement
compression protocol as the optimal rate at
which a measurement gathers information~\cite{Winter04}.
\end{itemize}This last advantage is the most compelling one from the perspective of
quantum information theory---one cannot really justify a measure as an information measure unless it corresponds
to a meaningful information processing task. Indeed, when reading the first few paragraphs of Groenewold's 
paper \cite{G71}, it becomes evident that his original motivation was information theoretic in nature,
and with this in mind, Winter's measure in \eqref{eq:winter-info-gain}
is clearly the one Groenewold was seeking after all.

In spite of the above arguments in favor of the information measure in \eqref{eq:winter-info-gain} as a measure
of information gain, it is still lacking in one aspect: it is dependent on the state on which the
quantum measurement $\mathcal{M}$ acts in addition to the
measurement itself. A final requirement that one should impose for a measure of information gain by a measurement
is that it should depend only on the measurement itself. A simple way to remedy this problem is
to maximize the quantity in \eqref{eq:winter-info-gain} over all possible input states, leading to the following
characterization of information gain:
\begin{equation}
I(\mathcal{M}) \equiv \max_{\rho_{AR}} I(X:R)_{\omega} \, , \label{eq:max-info-gain}
\end{equation}
for $\omega_{RX}$ as in \eqref{eq:info-gain-state}. 
The quantity
above has already been identified and studied by previous authors as an important information quantity,
being labeled as the ``purification capacity'' of a measurement \cite{J03,jacobs:012102} or
the ``information capacity of a quantum observable'' \cite{H11}. The above quantity also admits
an operational interpretation as the entanglement-assisted capacity of a quantum measurement
for transmitting classical information \cite{Bennett02,Hol01a,H11},
though it is our opinion that this particular operational interpretation is not sufficiently compelling
such that we should associate the measure in \eqref{eq:max-info-gain} with the notion of information gain.
The main aim of this paper is to address this issue by providing a compelling operational interpretation of
the measure in \eqref{eq:max-info-gain}.

% Should we cite the following? italians \cite{DDS11}. Oreshkov \cite{OCMB11}.
% Nielsen \cite{N01}. \cite{BL05}. 
 
%%%%%%%%%%%%%%%%%%%%%%%%%%%%%%%%%%%%%%%%%%%%%

\section{Summary of Results}

In this paper, our main contribution is to show that $I(\mathcal M)$
is the optimal rate at which a measurement gains information when many identical instances of
it act on an arbitrary input state. In our opinion, this new result establishes~\eqref{eq:max-info-gain}
as {\it the} information-theoretic measure of information gain of a quantum measurement.
In more detail, let $A$ denote the input Hilbert space for a
given measurement $\mathcal{M}$. We suppose that a third party prepares an arbitrary
quantum state on a Hilbert space $A^{\otimes n}$, which is equivalent to $n$ identical copies of the 
original Hilbert space $A$, where $n$ is a large positive number. A sender and receiver can then exploit
some amount of shared random bits and classical communication
to {\it simulate} the action of $n$ instances of the measurement
$\mathcal{M}$ (denoted by $\mathcal{M}^{\otimes n}$) on the chosen input state,
in such a way that it becomes physically impossible for the third
party, to whom the receiver passes along the measurement outcomes,
to distinguish between the simulation and the ideal measurement $\mathcal{M}^{\otimes n}$
as $n$ becomes large (the third party can even
keep the purifying system of a purification of the chosen input state in order to
help with the distinguishing task). By design,
the information gained by the measurement is that
relayed by the classical communication.  Following~\cite{Winter04},
we call this task {\it universal measurement compression}. We
prove that the optimal rate of
classical communication is equal to $I(\mathcal M)$,
if sufficient shared randomness is available.

\begin{figure}[ptb]
\begin{center}
\includegraphics[
width=6in]{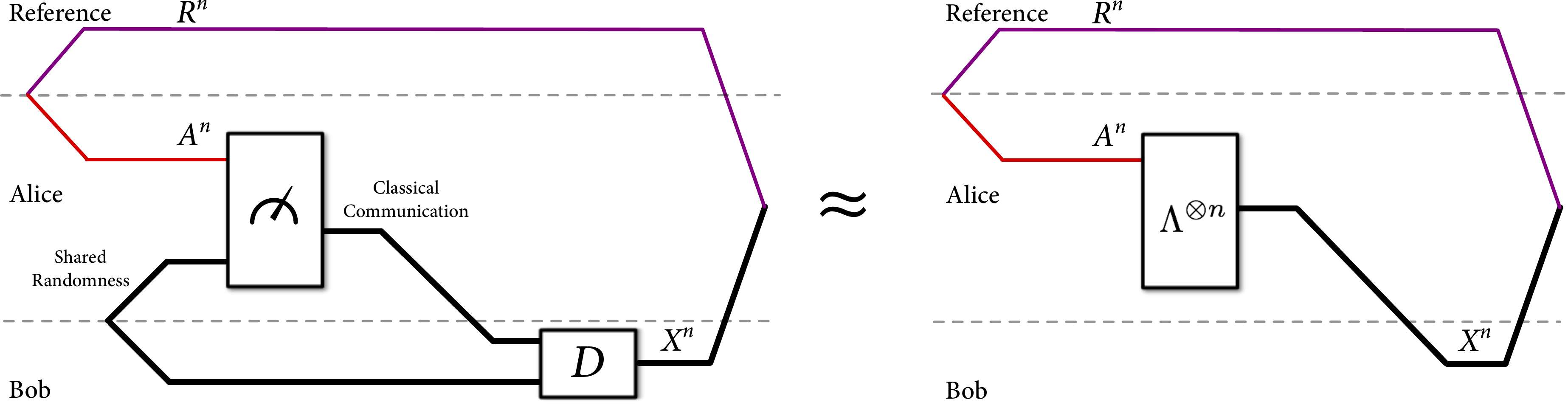}
\end{center}
\caption{Simulation (left) of the measurement $\Lambda^{\otimes n}$ (right). In the simulation, Alice uses shared randomness to perform a new measurement, whose result she communicates to Bob, such that Bob can recover the actual measurement output $X^n$ using the message and the shared randomness. If the simulation scheme works for any input, we can associate the amount of communication with the information gained by the measurement.}
\label{fig:measurement-compression}%
\end{figure}

The information-theoretic task outlined above is also known as {\it channel simulation} (depicted in Figure~\ref{fig:measurement-compression}),
and it has been well studied for the case of fully classical channels (with
classical inputs and classical outputs) \cite{Bennett02,C08,C12} and
fully quantum channels (with quantum inputs and quantum outputs)
\cite{Bennett02,Bennett06,Berta09_2}. The ``in-between'' case 
of channels with quantum inputs and classical outputs (i.e., measurements)
has been studied as well \cite{Winter04} (see also \cite{Wilde12}),
but as mentioned above, the problem of simulating many instances of a quantum measurement
on an {\it arbitrary} input state has not been studied before this paper. Beyond its intrinsic interest
as an information-processing task, channel simulation has two known concrete applications:
in establishing a strong converse rate for a channel coding task \cite{Bennett02,Bennett06,entCost,6284692}
and in rate distortion coding (lossy data compression) \cite{winter_new,DHW11,DHWW12,DWHW12}.

Our paper also features some related results of interest.
We characterize the optimal rate region consisting of the rates of shared randomness
and classical communication that are both necessary and sufficient for the existence of a measurement simulation,
whenever both the sender and receiver are required to obtain the measurement outcomes
(this is known as a feedback
simulation since the sender also obtains the measurement outcomes).
We also characterize
the optimal rate region of shared randomness and classical communication
for a non-feedback simulation,
in which the sender is not required to obtain the
measurement outcomes. Note that if sufficient shared
randomness is available and we are only interested in
quantifying the rate of classical communication,
then there is no advantage of a non-feedback simulation
over a feedback one---the optimal rate of
classical communication is given by \eqref{eq:max-info-gain}.

Our proof technique in this paper
exploits ideas from the approach in \cite{Berta09_2}
for proving the fully quantum reverse Shannon theorem. In fact, one can think of our approach
here as a ``classicalized'' or ``dephased'' version of that approach. In particular, we begin
by establishing a protocol known as ``classically coherent state merging,'' which
is a variation of the well-known state merging protocol \cite{Horodecki05,Horodecki06}
specialized to classically coherent states
(see Section~\ref{sec:prelim} for  definition). We then show how time-reversing this protocol and exchanging the roles
of Alice and Bob leads to a protocol known as ``classically coherent state splitting.''
It suffices for our purposes for this protocol
to use shared randomness
and classical communication rather than entanglement and quantum communication, respectively.
Generalizing this last protocol then leads to a one-shot state-and-channel simulation which is essentially optimal when acting on a single copy
of a \emph{known} state. Finally, we exploit the post-selection technique for quantum channels~\cite{ChristKoenRennerPostSelect}
and the aforementioned state splitting protocol
to show that it suffices to simulate
many instances of a measurement on a purification of a particular de Finetti quantum input state
in order to guarantee that the simulation is asymptotically
perfect when acting on an \emph{arbitrary} quantum state. We then show that
applying very similar reasoning as above along with randomness recycling \cite{Bennett06} 
solves the non-feedback case.

We organize this paper as follows. In Section~\ref{sec:prelim},
we introduce our notation and review
preliminary concepts such as states, distance measures, channels,
isometries, entropies, smooth entropies,
and classically coherent states. Section~\ref{sec:cl-coh-merge-split}
then introduces one-shot protocols for state merging and state
splitting of classically coherent states (the classical state
splitting protocol turns out to be the most important  tool
for proving our main result). Section~\ref{sec:main-result}
provides a proof of our main results for the case of feedback
and non-feedback simulations, and we shortly comment on possible extensions and
applications in Section~\ref{sec:applications}. We finally
conclude in Section~\ref{sec:conclusion} by summarizing our
results and stating some directions for future research.

\section{Preliminaries}

\label{sec:prelim}

\textbf{States, Distance Measures, Channels, Isometries}.
Let $A,B,C,\ldots$ denote finite dimensional Hilbert spaces and let $|A|$ denote
the dimension of~$A$. We establish notation for several sets:
$\cL(A)$ linear operators on $A$,
$\cP(A)$ non-negative linear operators on $A$,
$\cS_{\leq}(A)=\{\rho_{A}\in\cP(A)\,|\,\mathrm{tr}[\rho]\leq1\}$
subnormalized states on $A$,
$\cS(A)=\{\rho_{A}\in\cP(A)\,|\,\mathrm{tr}[\rho]=1\}$ density operators or states on $A$,
and $\cV(A)=\{\rho_{A}\in\cS(A)\,|\,\mathrm{tr}[\rho^2]=1\}$ pure-state density operators on $A$.
We define the purified distance $P(\rho_{A},\sigma_{A})=\sqrt{1-\bar{F}^{2}(\rho_{A},\sigma_{A})}$ for $\rho_{A},\sigma_{A}\in\cS_{\leq}(A)$, where $\bar{F}(\rho_{A},\sigma_{A})=F(\rho_{A},\sigma_{A})+\sqrt{(1-\mathrm{tr}[\rho_{A}])(1-\mathrm{tr}[\sigma_{A}])}$, and the quantum fidelity $F(\rho_{A},\sigma_{A})=\|\sqrt{\rho_{A}}\sqrt{\sigma_{A}}\|_{1}$ with $\|\Gamma_{A}\|_{1}=\mathrm{tr}[\sqrt{\Gamma_{A}\Gamma_{A}^{\dagger}}]$ for $\Gamma_{A}\in\cL(A)$.
We use the notation $\rho_{A}\approx_{\eps}\sigma_{A}$ to indicate that $\rho_{A}$ and $\sigma_{A}$
are $\eps$-close with respect to the purified distance: $P(\rho_{A},\sigma_{A})\leq\eps$.
We define the $\eps$-ball around $\rho_{A}$ as $\cB^{\eps}(\rho_{A})=\{{\tilde\rho_{A}}\in\cS_{\leq}(A):\tilde{\rho}_{A}\approx_{\eps}\rho_{A}\}$. The tensor product of two Hilbert spaces $A$ and $B$ is denoted by $AB\equiv A\otimes B$. Given a multipartite operator $\rho_{AB}\in\cP(AB)$, we unambiguously
write $\rho_{A}=\mathrm{tr}_{B}[\rho_{AB}]$ for the corresponding reduced operator. For $M_{A}\in\cL(A)$, we write $M_{A}\equiv M_{A}\otimes\id_{B}$ for the enlargement on any joint Hilbert space $AB$, where $\id_{B}$ denotes the identity operator acting on $\cL(B)$.
Isometries from $A$ to $B$ are denoted by $V_{A\rightarrow B}$. For Hilbert spaces $A$, $B$ with orthonormal bases $\{\ket{i}_{A}\}_{i=1}^{|A|}$, $\{\ket{i}_{B}\}_{i=1}^{|B|}$ and $|A|=|B|$, the canonical identity mapping from $\cL(A)$ to $\cL(B)$ with respect to these bases is denoted by $\cI_{A\rightarrow B}$, i.e.,~$\cI_{A\rightarrow B}(\ket{i}\bra{j}_{A})=\ket{i}\bra{j}_{B}$. A linear map $\cE_{A\rightarrow B}:\cL(A)\rightarrow\cL(B)$ is positive if $\cE_{A\rightarrow B}(\rho_{A})\in\cP(B)$ for all $\rho_{A}\in\cP(A)$. It is completely positive if the map $(\cE_{A\rightarrow B}\otimes\cI_{C\rightarrow C})$ is positive for all $C$. Completely positive and trace preserving maps are called quantum channels. The support of $\rho_{A}\in\cP(A)$ is denoted by $\mathrm{supp}(\rho_{A})$, the projector onto $\mathrm{supp}(\rho_{A})$ is denoted by $\rho^{0}_{A}$ and $\mathrm{tr}\left[\rho^{0}_{A}\right]=\mathrm{rank}(\rho_{A})$, the rank of $\rho_{A}$. For $\rho_{A}\in\cP(A)$ we write $\|\rho_{A}\|_{\infty}$ for the operator norm of $\rho_{A}$, which is equal to the maximum eigenvalue of $\rho_{A}$.

\textbf{Diamond Norm}. We will need a distance measure for
quantum channels. We use a norm on the set of quantum
channels which measures the bias in distinguishing two
such mappings. In quantum information theory,
this norm is known as the diamond norm~\cite{Kitaev97}. Here,
we present it in a formulation which highlights that
it is dual to the well-known completely bounded
(cb) norm~\cite{Paulsen}. 

\begin{definition}\label{def:diamond}
Let $\cE_{A}:\cL(A)\mapsto\cL(B)$ be a linear map. The \textit{diamond norm} of $\cE_{A}$ is defined as
\begin{align}
\|\cE_{A}\|_{\diamond}=\sup_{k\in\mathbb{N}}\|\cE_{A}\otimes\cI_{k}\|_{1}\ ,
\end{align}
where $\|\cF_{A}\|_{1}=\sup_{\sigma\in\cS_{\leq}(A)}\|\cF_{A}(\sigma_{A})\|_{1}$ and $\cI_{k}$ denotes the identity map on states of a $k$-dimensional quantum system.
\end{definition}

The supremum in Definition~\ref{def:diamond} is reached for $k=|A|$~\cite{Kitaev97, Paulsen}. Two quantum channels $\cE$ and $\cF$ are called $\eps$-close if they are $\eps$-close in the metric induced by the diamond norm. 

\textbf{Classically Coherent States}. We say that a pure state $\ket{\psi}\bra{\psi}_{X_A X_B R} \in \cV(X_A X_B R)$ is {\it classically coherent} with respect to systems $X_A X_B$ if there is an orthonormal basis $\{ \ket{x} \}$ such that $ \ket{\psi}$ can be written in the following form:
\begin{equation}
	\label{eq:clacoh}
\ket{\psi}_{X_A X_B R} = \sum_x \sqrt{p_x} \,\ket{x x}_{X_A X_B} \otimes \ket{\psi_x}_R ,
\end{equation}
for some probability distribution $p_x$ and states $\vert \psi_x \rangle_R$. Harrow realized the importance of classically coherent states for quantum communication tasks \cite{Harrow04PRL}, while Refs.~\cite{Dupuis12,Szehr11} recently exploited this notion in devising a ``decoupling approach'' to the Holevo-Schumacher-Westmoreland coding theorem \cite{Holevo98,Schumacher97} that is useful for our purposes here. Classically coherent states
are also related to Zurek's approach to decoherence \cite{Z91}, in which classicality arises from an inaccessible environment possessing an ``imprint'' of a classical state in superposition (as in the above state if we think of
$X_B$ as an environment).

\textbf{Entropies}. Recall the following standard definitions. The von Neumann entropy of $\rho_{A}\in\cS(A)$ is defined as\footnote{All logarithms in this paper are taken to base 2.}
\begin{align}
H(A)_{\rho}=-\mathrm{tr}\left[\rho_{A}\log\rho_{A}\right]\ .
\end{align}
The quantum relative entropy of $\rho_{A}\in\cS_{\leq}(A)$ with respect to $\sigma_{A}\in\cP(A)$ is given by
\begin{align}
D(\rho_{A}\|\sigma_{A})=\mathrm{tr}[\rho_{A}\log\rho_{A}]-\mathrm{tr}[\rho_{A}\log\sigma_{A}] ,
\end{align}
if $\mathrm{supp}(\rho_{A})\subseteq\mathrm{supp}(\sigma_{A})$ and $\infty$ otherwise. The conditional von Neumann entropy of $A$ given $B$ for $\rho_{AB}\in\cS(AB)$ is defined as
\begin{align}
H(A|B)_{\rho}=-D(\rho_{AB}\|\id_{A}\otimes\rho_{B})\ .
\end{align}
The mutual information between $A$ and $B$ for $\rho_{AB}\in\cS(AB)$ is given by
\begin{align}
I(A:B)_{\rho}=D(\rho_{AB}\|\rho_{A}\otimes\rho_{B})\ .
\end{align}
Note that we can also write
\begin{align}
& H(A|B)_{\rho}=-\inf_{\sigma_{B}\in\cS(B)}D(\rho_{AB}\|\id_{A}\otimes\sigma_{B}) ,\\
& I(A:B)_{\rho}=\inf_{\sigma_{B}\in\cS(B)}D(\rho_{AB}\|\rho_{A}\otimes\sigma_{B})\ .
\end{align}

\textbf{Smooth Entropies}. We now give the definitions of the smooth entropy measures that we need in this work. We define the max-relative entropy of $\rho_{A}\in\cS_{\leq}(A)$ with respect to $\sigma_{A}\in\cP(A)$ as \cite{datta-2008-2}
\begin{align}
D_{\max}(\rho_{A}\|\sigma_{A})=\inf\{\lambda\in\mathbb{R}:2^{\lambda}\cdot\sigma_{A}\geq\rho_{A}\}\ .
\end{align}
The conditional min-entropy of $A$ given $B$ for $\rho_{AB}\in\cS_{\leq}(AB)$ is defined as
\begin{align}
H_{\min}(A|B)_{\rho}=-\inf_{\sigma_{B}\in\cS(B)}D_{\max}(\rho_{AB}\|\id_{A}\otimes\sigma_{B})\ .
\end{align}
In the special case where $B$ is trivial, we get $H_{\min}(A)_{\rho}=-\log\|\rho_{A}\|_{\infty}$. The max-information that $B$ has about $A$ for $\rho_{AB}\in\cS_{\leq}(AB)$ is defined as \cite{Berta09_2}
\begin{align}
I_{\max}(A:B)_{\rho}=\inf_{\sigma_{B}\in\cS(B)}D_{\max}(\rho_{AB}\|\rho_{A}\otimes\sigma_{B})\ .
\end{align}
Note that, unlike the mutual information, the max-information is not symmetric in its arguments.\footnote{For a further discussion of max-based measures for \textit{mutual information}, see~\cite{Ciganovic12}.}

Smooth entropy measures are defined by extremizing the non-smooth
measures over a set of nearby states, where our notion of ``nearby''
is expressed in terms of the purified distance.
The smooth max-information that $B$ has about $A$ for
$\rho_{AB}\in\cS_{\leq}(AB)$ is defined as
\begin{align}\label{eq:maxinfo}
I_{\max}^{\eps}(A:B)_{\rho}=\inf_{\bar{\rho}_{AB}\in\cB^{\eps}(\rho_{AB})}I_{\max}(A:B)_{\bar{\rho}}\ .
\end{align}
In contrast to the non-smooth case, the smooth max-information is approximately symmetric in its arguments.

\begin{lemma}\cite[Corollary 4.2.4]{Ciganovic12}\label{lem:ciganovic}
Let $\eps\geq0$, $\eps'>0$, and $\rho_{AB}\in\cS(AB)$. Then, we have that
\begin{align}
I_{\max}^{\eps+2\eps'}(B:A)_{\rho}\leq I_{\max}^{\eps}(A:B)_{\rho}+\log\left(\frac{2}{(\eps')^{2}}+2\right)\ ,
\end{align}
and the same holds for $A$ and $B$ interchanged.
\end{lemma}

For technical reasons, we also need the following entropic quantities. For $\eps\geq0$, and $\rho_{A}\in\cS_{\leq}(A)$, the max-entropy and its smooth version are defined as
\begin{align}
&H_{\max}(A)_{\rho}=2\log\mathrm{tr}\left[\rho_{A}^{1/2}\right],\\
&H_{\max}^{\eps}(A)_{\rho}=\inf_{\bar{\rho}_{A}\in\cB^{\eps}(\rho_{A})}H_{\max}(A)_{\bar{\rho}}\ .
\end{align}
Furthermore, the zero-R\'enyi entropy and its smooth version are defined as
\begin{align}
&H_{0}(A)_{\rho}=\log\mathrm{rank}(\rho_{A}),\\
&H_{0}^{\eps}(A)_{\rho}=\inf_{\bar{\rho}_{A}\in\cB^{\eps}(\rho_{A})}H_{0}(A)_{\bar{\rho}}\ .
\end{align}

Since all Hilbert spaces in this paper are assumed to be
finite dimensional and the ball $\cB^{\eps}$ is convex and compact \cite{tomamichel:thesis},
we can replace the infima by minima and the suprema by maxima
in all the definitions of this section. We will do so in what follows.

%%%%%%%%%%%%%%%%%%%%%%%%%%%%%%%%%%%%%%%%%%%%%

\section{Classically Coherent State Merging and State Splitting}

\label{sec:cl-coh-merge-split}

We first establish ``one-shot'' protocols for state merging and
state splitting of classically coherent quantum states.
The classical state splitting protocol established in this section will then be the basis
for the universal measurement compression protocol discussed in the next section.

\begin{definition}[State Merging for Classically Coherent States]
Consider a bipartite system with parties Alice and Bob. Let $\eps>0$, and $\rho_{X_{A}X_{B}BR}\in\cV(X_{A}X_{B}BR)$ be classically coherent on $X_{A}X_{B}$ with respect to the basis $\{\ket{x}\}$, where Alice controls $X_{A}$, Bob $X_{B}B$, and $R$ is a reference system. A quantum protocol $\cE$ is called an $\eps$-error state merging of $\rho_{X_{A}X_{B}BR}$ if it consists of applying local operations at Alice's side, sending $q$ qubits from Alice to Bob, local operations at Bob's side, and it outputs a state $\omega_{X_{B'}X_{B}BRX_{A_{1}}B_{1}}=(\cE\otimes\cI_{R})(\rho_{X_{A}X_{B}BR})$ such that
\begin{align}
\omega_{X_{B'}X_{B}BRX_{A_{1}}B_{1}}\approx_{\eps}\cI_{X_{A}\rightarrow X_{B'}}(\rho_{X_{A}X_{B}BR})\otimes\phi^{E}_{X_{A_{1}}B_{1}}\ ,
\end{align}
where $\phi^{E}_{X_{A_{1}}B_{1}}$ is a maximally entangled state of Schmidt rank $E$. The quantity $q$ is called the quantum communication cost, and $e=\lfloor\log E\rfloor$ the entanglement gain.
\end{definition}

\begin{lemma}\label{lem:merging}
Let $\eps>0$, and $\rho_{X_{A}X_{B}BR}\in\cV(X_{A}X_{B}BR)$ be classically coherent on $X_{A}X_{B}$ with respect to the basis $\{\ket{x}\}$. Then there exists an $\eps$-error state merging protocol for $\rho_{X_{A}X_{B}BR}$ with quantum communication cost
\begin{align}
q=\left\lceil H_{0}(X_{A})_{\rho}-H_{\min}(X_{A}|R)_{\rho}+4\cdot\log\frac{1}{\eps}  \right\rceil ,
\end{align}
and entanglement gain
\begin{align}
e=\left\lfloor H_{\min}(X_{A}|R)_{\rho}-4\cdot\log\frac{1}{\eps}\right\rfloor\ .
\end{align}
\end{lemma}

\begin{figure}[ptb]
\begin{center}
\includegraphics[
width=4.5in]{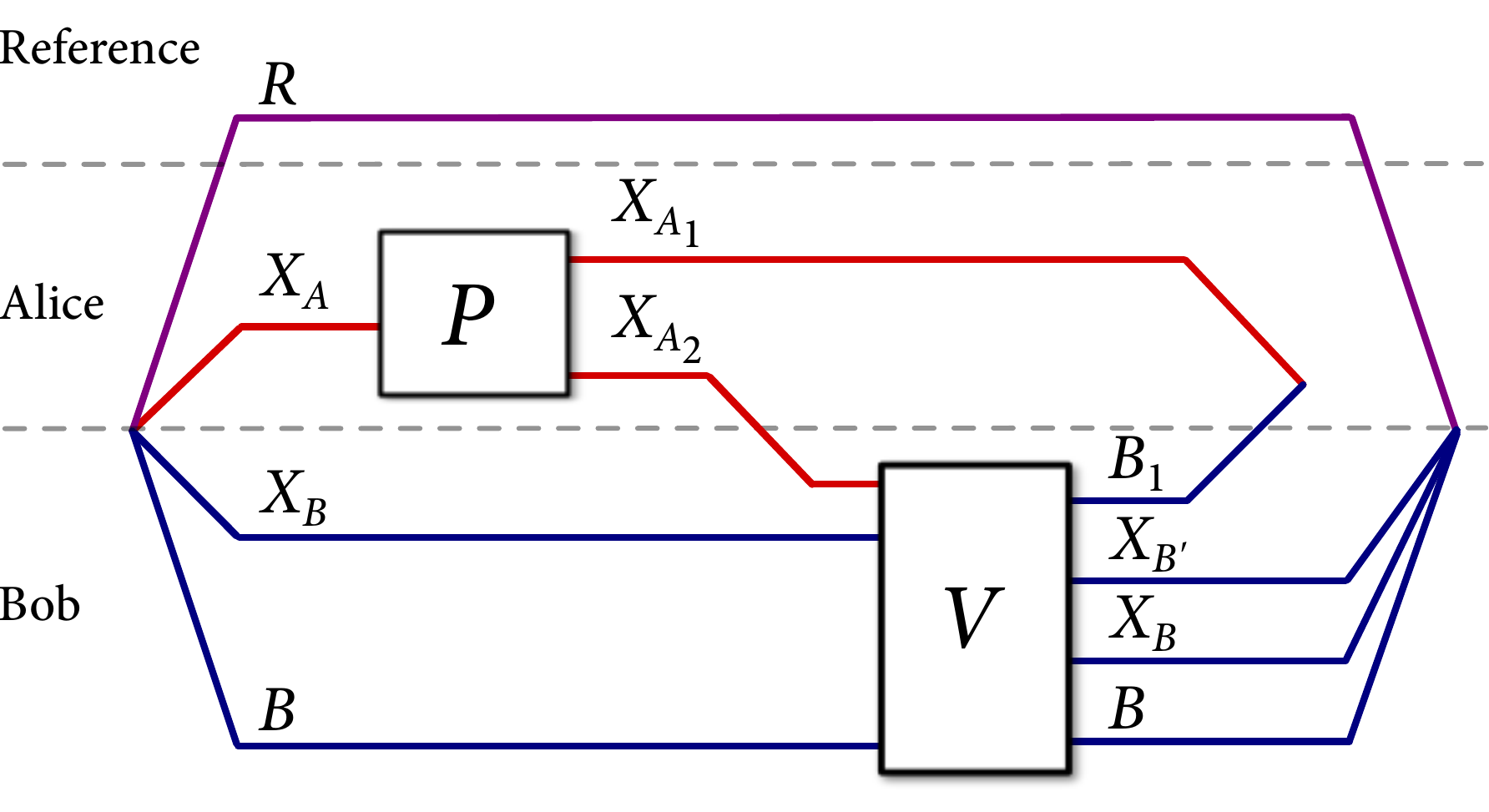}
\end{center}
\caption{The protocol from the proof of Lemma~\ref{lem:merging}
for state merging of a classically coherent state on systems $R X_{A}X_{B}B$.
The operation $P$ is a permutation of states in the orthonormal basis $\{ \ket{x} \}$ of $X_{A}$, and it also
splits $X_{A}$ into two subsystems. The operation $V$ is an isometry
guaranteed by Uhlmann's theorem to complete the merging task, while also generating entanglement between
Alice and Bob.}
\label{fig:state-merging-classically-coherent}%
\end{figure}

\begin{proof}
The intuition is as follows. First Alice applies a particular permutation $P_{X_{A}\rightarrow X_{A_{1}}X_{A_{2}}}$ in the basis $\{\ket{x}\}_{x\in X_{A}}$; it also splits the output into two subsystems $X_{A_1}$ and $X_{A_2}$. Then she sends $X_{A_{2}}$ to Bob, who finally performs a local isometry $V_{X_{A_{2}}X_{B}B\rightarrow X_{B'}X_{B}BB_{1}}$. After Alice applies the permutation, the state on $X_{A_{1}}R$ is approximately given by $\frac{\id_{X_{A_{1}}}}{|X_{A_{1}}|}\otimes\rho_{R}$ and Bob holds a purification of this. But $\frac{\id_{X_{A_{1}}}}{|X_{A_{1}}|}\otimes\rho_{R}$ is the reduced state of $\rho_{X_{B'}X_{B}BR}\otimes\phi^{E}_{X_{A_{1}}B_{1}}$, and since all purifications are equivalent up to local isometries, there exists an isometry $V_{X_{A_{2}}X_{B}B\rightarrow X_{B'}X_{B}BB_{1}}$ on Bob's side that transforms the state into $\rho_{X_{B'}X_{B}BR}\otimes\phi^{E}_{X_{A_{1}}B_{1}}$. Figure~\ref{fig:state-merging-classically-coherent} depicts this protocol.

More formally, let $X_{A}=X_{A_{1}}X_{A_{2}}$ with $\log|X_{A_{2}}|=\lceil\log|X_{A}|-H_{\min}(X_{A}|R)_{\rho}+4\cdot\log\frac{1}{\eps}\rceil$. According to Proposition~\ref{prop:extractor} concerning permutation based extractors, there exists a permutation $P_{X_{A}\rightarrow X_{A_{1}}X_{A_{2}}}$ such that for $\sigma_{X_{A_{1}}X_{A_{2}}BR}=P_{X_{A}\rightarrow X_{A_{1}}X_{A_{2}}}(\rho_{X_{A}X_{B}BR})$,
\begin{align}\label{eq:random}
\left\|\sigma_{X_{A_{1}}R}-\frac{\id_{X_{A_{1}}}}{|X_{A_{1}}|}\otimes\rho_{R}\right\|_{1}\leq\eps^{2}\ .
\end{align}
By an upper bound of the purified distance in terms of the trace distance (Lemma~\ref{lem:distance}), this implies $\sigma_{X_{A_{1}}R}\approx_{\eps}\frac{\id_{X_{A_{1}}}}{|X_{A_{1}}|}\otimes\rho_{R}$. Alice applies this permutation $P_{X_{A}\rightarrow X_{A_{1}}X_{A_{2}}}$ and then sends $X_{A_{2}}$ to Bob; therefore
\begin{align}
q=\left\lceil\log|X_{A}|-H_{\min}(X_{A}|R)_{\rho}+4\cdot\log\frac{1}{\eps}\right\rceil\ .
\end{align}
Uhlmann's theorem~\cite{uhlmann,Jozsa94} guarantees that there exists an isometry $V_{X_{A_{2}}X_{B}B\rightarrow X_{B'}X_{B}BB_{1}}$ such that
\begin{align}\label{eq:uhlmann}
P\left(\sigma_{X_{A_{1}}R},\,\frac{\id_{X_{A_{1}}}}{|X_{A_{1}}|}\otimes\rho_{R}\right)=P\left(V_{X_{A_{2}}X_{B}B\rightarrow X_{B'}X_{B}BB_{1}}(\sigma_{X_{A_{1}}X_{A_{2}}X_{B}BR}), \, \phi^{E}_{X_{A_{1}}B_{1}}\otimes\rho_{X_{B'}X_{B}BR}\right)\ .
\end{align}
Hence the entanglement gain is given by
\begin{align}
e=\left\lfloor H_{\min}(X_{A}|R)_{\rho}-4\cdot\log\frac{1}{\eps}\right\rfloor\ .
\end{align}
Now if $\rho_{X_{A}}$ has full rank, this is already what we want. In general $\log\mathrm{tr}\left[\rho_{X_{A}}^{0}\right]=\log|X_{\hat{A}}|\leq\log|X_{A}|$. But in this case we can restrict $X_{A}$ to the subspace $X_{\hat{A}}$ on which $\rho_{X_{A}}$ has full rank, i.e.\ those $x$ for which $p_x\neq 0$.
\end{proof}

\begin{definition}[State Splitting for Classically Coherent States]\label{def:splitting}
Consider a bipartite scenario with parties Alice and Bob. Let $\eps>0$, and $\rho_{AX_{A}X_{A'}R}\in\cV(AX_{A}X_{A'}R)$ be classically coherent on $X_{A}X_{A'}$ with respect to the basis $\{\ket{x}\}$, where Alice controls $AX_{A}X_{A'}$, and $R$ is a reference system. Furthermore let $\phi^{E}_{A_{1}B_{1}}$ be a maximally entangled state of Schmidt rank $E$ shared between Alice and Bob. A quantum protocol $\cE$ is called an $\eps$-error state splitting of $\rho_{AX_{A}X_{A'}R}$ if it consists of applying local operations at Alice's side, sending $q$ qubits from Alice to Bob, local operations at Bob's side, and it outputs a state $\omega_{AX_{A}X_{B}R}=(\cE\otimes\cI_{R})(\rho_{AX_{A}X_{A'}R}\otimes\phi^{E}_{A_{1}B_{1}})$ such that
\begin{align}
\omega_{AX_{A}X_{B}R}\approx_{\eps}\cI_{X_{A'}\rightarrow X_{B}}(\rho_{AX_{A}X_{A'}R})\ .
\end{align}
The quantity $q$ is called the quantum communication cost, and $e=\lceil\log E\rceil$ the entanglement cost.
\end{definition}

\begin{lemma}\label{lem:splitting}
Let $\eps>0$, and $\rho_{AX_{A}X_{A'}R}\in\cV(AX_{A}X_{A'}R)$ be classically coherent on $X_{A}X_{A'}$ with respect to the basis $\{\ket{x}\}$. Then there exists an $\eps$-error state splitting protocol for $\rho_{AX_{A}X_{A'}R}$ with quantum communication cost
\begin{align}
q=\left\lceil H_{0}(X_{A'})_{\rho}-H_{\min}(X_{A'}|R)_{\rho}+4\cdot\log\frac{1}{\eps}\right\rceil ,
\end{align}
and entanglement cost
\begin{align}
e=\left\lfloor H_{\min}(X_{A'}|R)_{\rho}-4\cdot\log\frac{1}{\eps}\right\rfloor\ .
\end{align}
\end{lemma}

\begin{proof}
We get the desired state splitting protocol by time-reversing the state merging protocol of Lemma~\ref{lem:merging} and interchanging the roles of Alice and Bob.
Figure~\ref{fig:state-splitting-classically-coherent}(a) depicts the state splitting protocol for
classically coherent states. More precisely, we first define an isometry $V_{X_{A'_{2}}X_{A}A\rightarrow X_{A'}X_{A}AA_{1}}$, analogously to $V_{X_{A_{2}}X_{B}B\rightarrow X_{B'}X_{B}BB_{1}}$ of~\eqref{eq:uhlmann} in the state merging protocol. Because all isometries are injective, we can define an inverse of $V$ acting on the image of $V$ (which we denote by $\mathrm{Im}(V)$). The inverse is again an isometry and we denote it by $V^{-1}_{\mathrm{Im}(V)\rightarrow X_{A'_{2}}X_{A}A}$. The protocol starts by measuring the $AX_AX_{A'}A_1$ systems to 
%applying a quantum operation to the state $\rho_{AX_{A}X_{A'}R}\otimes\phi^{E}_{A_{1}B_{1}}$, that first does a measurement on $\rho_{AX_{A}X_{A'}}\otimes\phi^{E}_{A_{1}}$ to 
decide whether $\rho_{AX_{A}X_{A'}}\otimes\phi^{E}_{A_{1}}\in\mathrm{Im}(V)$ or not. If so, the protocol proceeds by applying the isometry $V^{-1}_{\mathrm{Im}(W)\rightarrow X_{A'_{2}}X_{A}A}$, but otherwise the state is discarded and replaced with $\proj{0}_{X_{A'_{2}}X_{A}A}$. This step is necessary because the output of merging is not exactly $\rho_{AX_AX_{A'}R}$. 
 The next step is to send $X_{A'_{2}}$ to Bob, who then applies the permutation $P^{-1}_{X_{A'_{2}}B_{1}\rightarrow X_{B}}$ defined analogously to $P_{X_{A}\rightarrow X_{A_{1}}X_{A_{2}}}$ in~\eqref{eq:random}. By the monotonicity of the purified distance, we get a state that is $\eps$-close to $\cI_{X_{A'}\rightarrow X_{B}}(\rho_{AX_{A}X_{A'}R})$. \end{proof}

\begin{figure}[ptb]
\begin{center}
\includegraphics[
width=6.5in]{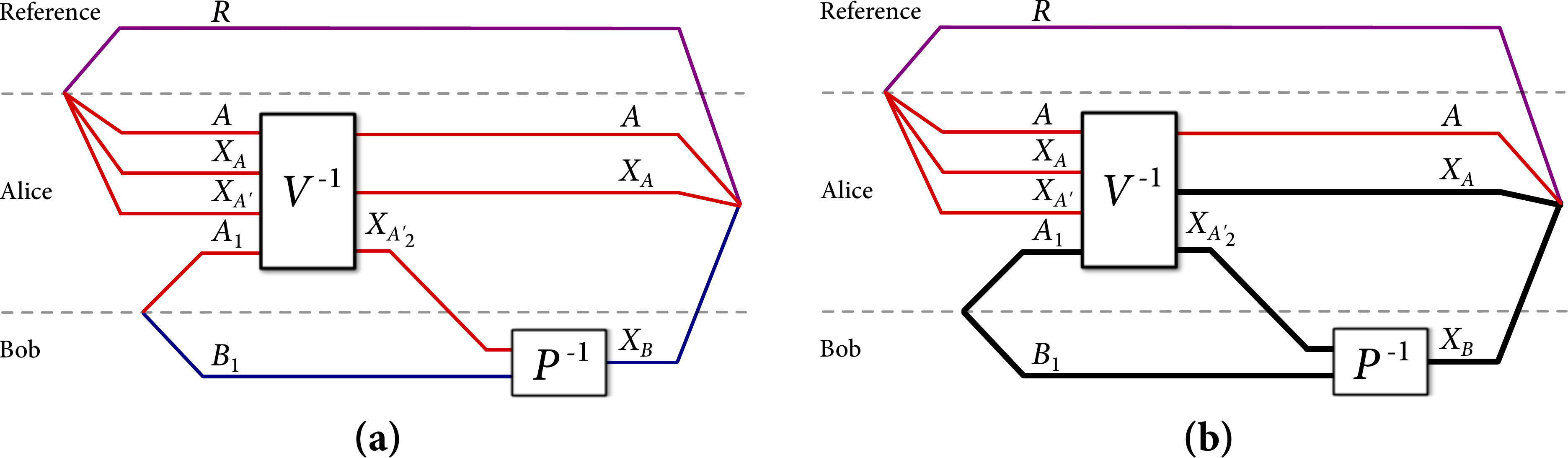}
\end{center}
\caption{\textbf{(a)} A simple protocol for state splitting obtained by time-reversing the state
merging protocol of Lemma~\ref{lem:merging} and interchanging the roles of Alice and Bob.
\textbf{(b)} If it is not necessary to maintain the quantum coherence
of the $X$ systems (if they can be dephased to
classical registers), then the state splitting protocol can exploit shared randomness and classical
communication instead of entanglement and quantum communication, respectively.}
\label{fig:state-splitting-classically-coherent}%
\end{figure}

%For the same premises as in the definition of state splitting of classically coherent states (Definition~\ref{def:splitting}), but 
If we are not concerned with the coherence of the registers $X_{A}$ and $X_{B}$ shared between Alice and Bob, then the protocol given above (Lemma~\ref{lem:splitting}) also works if the entanglement assistance and the quantum communication are replaced by the same amount of shared randomness assistance and classical communication, respectively. More precisely, we define:

\begin{definition}[Classical State Splitting of Classically Coherent States]\label{def:csplitting}
Consider a bipartite system with parties Alice and Bob. Let $\eps>0$, and $\rho_{AX_{A}X_{A'}R}\in\cV(AX_{A}X_{A'}R)$ be classically coherent on $X_{A}X_{A'}$ with respect to the basis $\{\ket{x}\}$, where Alice controls $AX_{A}X_{A'}$, and $R$ is a reference system. Furthermore let $\overline{\phi}^{S}_{X_{A_{1}}X_{B_{1}}}$ denote $S$ bits of shared randomness shared between Alice and Bob. A quantum protocol $\cE$ is called an $\eps$-error classical state splitting of $\rho_{AX_{A}X_{A'}R}$ if it consists of applying local operations at Alice's side, sending $c$ bits from Alice to Bob, local operations at Bob's side, and it outputs a state $\omega_{AX_{A}X_{B}R}=(\cE\otimes\cI_{R})(\rho_{AX_{A}X_{A'}R}\otimes\overline{\phi}^{S}_{X_{A_{1}}X_{B_{1}}})$ such that
\begin{align}
\omega_{AX_{A}X_{B}R}\approx_{\eps}\sum_{x}\bra{x}\rho_{AX_{A}X_{A'}R}\ket{x}_{X_{A'}}\otimes\proj{x}_{X_{B}}\ .
\end{align}
The quantity $c$ is called the classical communication cost, and $s=\lceil\log S\rceil$ shared randomness cost.
\end{definition}

Using the achievability of state splitting of classically coherent states (Lemma~\ref{lem:splitting}) we get the following.

\begin{corollary}\label{cor:csplitting}
Let $\eps>0$, and $\rho_{AX_{A}X_{A'}R}\in\cV(AX_{A}X_{A'}R)$ be classically coherent on $X_{A}X_{A'}$ with respect to the basis $\{\ket{x}\}$. Then there exists a classical $\eps$-error state splitting protocol for $\rho_{AX_{A}X_{A'}R}$ with classical communication cost
\begin{align}
c=\left\lceil H_{0}(X_{A'})_{\rho}-H_{\min}(X_{A'}|R)_{\rho}+4\cdot\log\frac{1}{\eps}\right\rceil
\end{align}
and shared randomness cost
\begin{align}
s=\left\lfloor H_{\min}(X_{A'}|R)_{\rho}-4\cdot\log\frac{1}{\eps}\right\rfloor\ .
\end{align}
\end{corollary}

\begin{proof}
Note that it is sufficient to find a protocol for state splitting of classically coherent states (as in Definition~\ref{def:splitting}) that only works up to random phase flips on the $X_{B}$ register. These random phase flips then commute with the action of the permutation that takes systems $B_1$ and $X_{A'_2}$ to $X_B$. Thus, if we use the protocol for state splitting of classically coherent states described before (Lemma~\ref{lem:splitting}), random phase flips on $X_{B}$ are the same as random phase flips on $X_{A'_{2}}B_{1}$ before the permutation $P^{-1}_{X_{A'_{2}}B_{1}\rightarrow X_{B}}$ is applied. Since random phase flips on $B_{1}$ just transform the maximally entangled state $\phi_{A_{1}B_{1}}$ to shared randomness $\overline{\phi}_{X_{A_{1}}X_{B_{1}}}$ of the same size (with the relabeling of $A_{1}B_{1}$ to $X_{A_{1}}X_{B_{1}}$), and they dephase the quantum system $X_{A'_{2}}$ to a classical system, the protocol of Lemma~\ref{lem:splitting} also works for classical state splitting of classically coherent states.
\end{proof}

Note that the above idea is similar to how Hsieh {\it et al.}~recovered the Holevo-Schumacher-Westmoreland coding theorem for classical communication from a protocol for entanglement-assisted classical communication \cite{hsieh08}, simply by dephasing shared entanglement to common randomness and replacing random unitaries with random permutations.

However, the classical communication cost of this
protocol is not yet optimal (for the general one-shot case
considered here). To improve this, we use an idea from a recent
proof of the quantum reverse Shannon theorem, and
Theorem~\ref{thm:converse} demonstrates that the rate
found in terms of the smooth max-information is
essentially optimal. The following lemma is the
crucial ingredient for the proof of our main result:
universal measurement compression (Theorem~\ref{thm:main}).

\begin{theorem}\label{thm:splitting}
Let $\eps>0$, $\eps'\geq0$, and $\rho_{AX_{A}X_{A'}R}\in\cV(AX_{A}X_{A'}R)$ be classically coherent on $X_{A}X_{A'}$ with respect to the basis $\{\ket{x}\}$. Then there exists a classical $(\eps+\eps'+\sqrt{8\eps'}+|X_{A'}|^{-1/2})$-error state splitting protocol for $\rho_{AX_{A}X_{A'}R}$ with
\begin{align}\label{eq:msplitting}
&c\leq I_{\max}^{\eps'}(X_{A'}:R)_{\rho}+4\cdot\log\frac{1}{\eps}+4+\log\log|X_{A'}|\\
&c+s\leq H_{0}^{\eps'}(X_{A'})_{\rho}+2+\log\log|X_{A'}|\ ,
\end{align}
where $c$ denotes the classical communication cost, and $s$ the shared randomness cost.
\end{theorem}

\begin{proof}
The idea for the protocol is as follows. Let $\rho_{AX_{A}X_{A'}R}=\proj{\rho}_{AX_{A}X_{A'}R}$ with
\begin{align}
\ket{\rho}_{AX_{A}X_{A'}R}=\sum_{x}\sqrt{p_{x}}\cdot\ket{xx}_{X_{A}X_{A'}}\otimes\ket{\rho^{x}}_{AR}\ .
\end{align}
First, in our proof, we disregard all the $x$ with $p_{x}\leq|X_{A'}|^{-2}$.
This introduces an error $|X_{A'}|^{-1/2}$, but the error at the end of the protocol
is still upper bounded by $|X_{A'}|^{-1/2}$ due to the
monotonicity of the purified distance. As the next step, we let Alice perform
a measurement $W_{X_{A'}\rightarrow X_{A'}Y_{A}}$ with
roughly $2\cdot\log|X_{A'}|$ measurement outcomes in the
basis $\{\ket{x}\}_{x\in X_{A'}}$. That is, the state after the measurement is of the form
\begin{align}
\omega_{AX_{A}X_{A'}RY_{A}}=\sum_{y}q_{y}\cdot\rho^{y}_{AX_{A}X_{A'}R}\otimes\proj{y}_{Y_{A}}\ ,
\end{align}
where the index $y$ indicates which measurement outcome occurs, $q_{y}$ denotes its probability, and $\rho^{y}_{AX_{A}X_{A'}R}$ is the corresponding post-measurement state. Then conditioned on the index $y$, we use the classical state splitting protocol for classically coherent states from Lemma~\ref{cor:csplitting} for each state $\rho^{y}_{AX_{A}X_{A'}R}$, and denote the corresponding classical communication cost and shared randomness cost by $c_{y}$ and $s_{y}$, respectively. The total amount of classical communication we need for this is no larger than $\max_{y}c_{y}$, plus the amount needed to send the register $Y_{A}$ (which is of order $\log\log|X_{A'}|$).
The sum cost is no larger than $\max_{y}c_{y} + s_{y}$ (along with the amount for sending $Y_{A}$). This completes the description of the classical state splitting protocol for $\rho_{AX_{A}X_{A'}R}$. All that remains to do is to bring the expression for the classical communication cost and the sum cost into the right form. In the following, we describe the proof in detail.

Let $Q=\lceil2\cdot\log|X_{A'}|-1\rceil$, $Y=\{0,1,\ldots,Q,(Q+1)\}$ and let $\{T_{X_{A'}}^{y}\}_{y\in Y}$ be a collection of projectors on $X_{A'}$ defined as
\begin{align}
T_{X_{A'}}^{Q+1}=\sum_{\substack{x\\0\leq p_{x}\leq2^{-2\log|X_{A'}|}}}\proj{x}_{X_{A'}}\ ,\qquad T_{X_{A'}}^{Q}=\sum_{\substack{x\\ 2^{-2\log|X_{A'}|} \leq p_{x}\leq2^{-Q} }}\proj{x}_{X_{A'}}\ ,
\end{align}
and for $y=0,1,\dots,(Q-1)$ as
\begin{align}
T_{X_{A'}}^{y}=\sum_{\substack{x\\2^{-(y+1)}\leq p_{x}\leq2^{-y}}}\proj{x}_{X_{A'}}\ .
\end{align}
These define a measurement
\begin{align}\label{eq:preproc}
W_{X_{A'}\rightarrow X_{A'}Y_{A}}(\cdot)=\sum_{y\in Y}T_{X_{A'}}^{y}(\cdot)T_{X_{A'}}^{y}\otimes\proj{y}_{Y_{A}}\ ,
\end{align}
where the vectors $\ket{y}_{Y_{A}}$ form an orthonormal basis, and $Y_{A}$ is at Alice's side. Furthermore let \begin{align}
q_{y} & =\mathrm{tr}\left[T_{X_{A'}}^{y}\rho_{X_{A'}}\right], \\
\rho_{AX_{A}X_{A'}R}^{y} & =q_{y}^{-1}\cdot T_{X_{A'}}^{y}\rho_{AX_{A}X_{A'}R}T_{X_{A'}}^{y},
\end{align}
and define the sub-normalized state
\begin{align}
\bar{\rho}_{AX_{A}X_{A'}R}=\sum_{y=0}^{Q}q_{y}\cdot\rho^{y}_{AX_{A}X_{A'}R}\ .
\end{align}
We have
\begin{align}\label{eq:cutoff}
P(\bar{\rho}_{AX_{A}X_{A'}R},\rho_{AX_{A}X_{A'}R})&=\sqrt{1-F^{2}(\bar{\rho}_{AX_{A}X_{A'}R},\rho_{AX_{A}X_{A'}R})}%=\sqrt{1-|\langle\bar{\rho}|\rho\rangle_{AX_{A}X_{A'}R}|^{2}}
\\
& \leq \sqrt{1-\sum_{y=0}^{Q}q_{y}}=\sqrt{q_{Q+1}}\leq\sqrt{|X_{A'}|\cdot2^{-2\log|X_{A'}|}}=|X_{A'}|^{-1/2}\ .
\end{align}
We proceed by defining the operations that we need for the classical state splitting protocol for $\bar{\rho}_{AX_{A}X_{A'}R}$. We want to use the $\eps$-error classical state splitting protocol from Corollary~\ref{cor:csplitting} for each $\rho^{y}_{AX_{A}X_{A'}R}$. For $y=0,1,\ldots,Q$ this protocol has a classical communication cost
\begin{align}
c_{y}\leq H_{0}(X_{A'})_{\rho^{y}}-H_{\min}(X_{A'}|R)_{\rho^{y}}+4\cdot\log\frac{1}{\eps}+1\ ,
\end{align}
and sum cost
\begin{align}\label{eq:sumcost}
c_{y}+s_{y}\leq H_{0}(X_{A'})_{\rho^{y}}\ ,
\end{align}
where $s_{y}$ denotes the shared randomness cost.

For $X_{A_{1}}$ on Alice's side, $X_{B_{1}}$ on Bob's side, and $X_{A_{1}^{y}}$, $X_{B_{1}^{y}}$ $2^{s_{y}}$-dimensional subspaces of $X_{A_{1}}$, $X_{B_{1}}$ respectively, the classical state splitting protocol from Corollary~\ref{cor:csplitting} has basically the following form: apply some isometry $V_{AX_{A'}X_{A}X_{A_{1}^{y}}\rightarrow AX_{(A'_{2})^{y}}X_{A}}$ on Alice's side, send $X_{(A'_{2})^{y}}$ from Alice to Bob (relabel it to $X_{B_{2}^{y}}$), and then apply some isometry $U_{X_{B_{1}^{y}}X_{B_{2}^{y}}\rightarrow B}$ on Bob's side
($U_{X_{B_{1}^{y}}X_{B_{2}^{y}}\rightarrow B}$ is the inverse permutation
discussed in the proof of Corollary~\ref{cor:csplitting}). As the next ingredient, we define the operations that supply the shared randomness of size $s_{y}$. For $y=0,1,\ldots,Q$, let $S_{X_{A_{1}^{y}}}$ and $S_{X_{B_{1}^{y}}}$ be the local operations at Alice's and Bob's side respectively, that put shared randomness of size $s_{y}$ on $X_{A_{1}^{y}}X_{B_{1}^{y}}$.

\begin{figure}[ptb]
\begin{center}
\includegraphics[
width=6.5in]{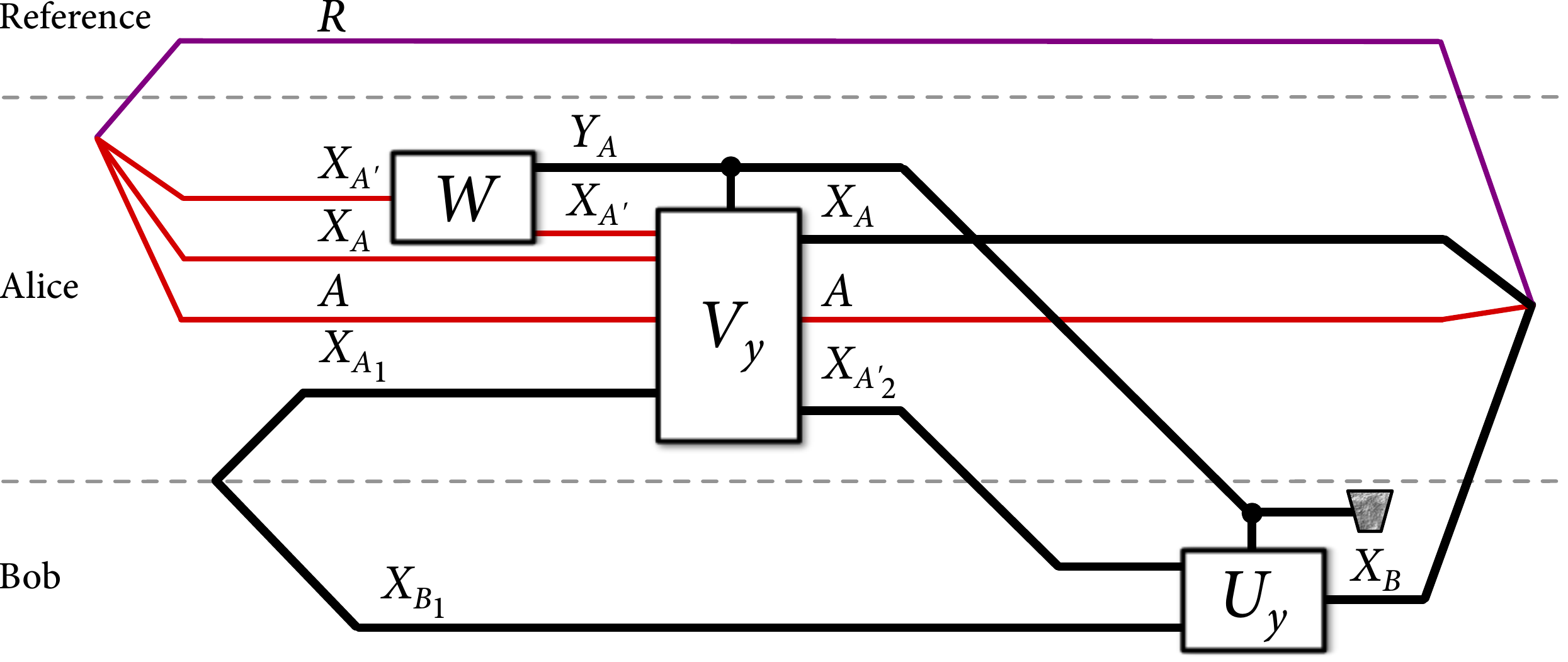}
\end{center}
\caption{Our final one-shot protocol for state splitting that achieves the smooth max-information rate
of Theorem~\ref{thm:splitting}. The converse theorem in Theorem~\ref{thm:converse}
guarantees that this one-shot protocol
is essentially optimal in terms of its classical communication cost.}
\label{fig:state-splitting-classically-coherent-max-info}%
\end{figure}

We are now ready to put the steps together and give the protocol for classical state splitting of $\bar{\rho}_{AX_{A}X_{A'}R}$ (depicted in Figure~\ref{fig:state-splitting-classically-coherent-max-info}).
Alice applies the measurement $W_{X_{A'}\rightarrow X_{A'}Y_{A}}$ from~\eqref{eq:preproc} followed by
\begin{align}
S_{A_{1}Y_{A}}=\sum_{y=1}^{Q}S_{X_{A_{1}^{y}}}\otimes\proj{y}_{Y_{A}}\ ,
\end{align}
and the isometry 
\begin{align}
V_{AX_{A}X_{A'}X_{A_{1}}Y_{A}\rightarrow AX_{A'_{2}}X_{A}Y_{A}}=\sum_{y=0}^{Q}V_{AX_{A}X_{A'}X_{A_{1}^{y}}\rightarrow AX_{(A'_{2})^{y}}X_{A}}\otimes\proj{y}_{Y_{A}}\ .
\end{align}
Afterwards she sends $X_{A'_{2}}$ and $Y_{A}$, that is
\begin{align}\label{eq:ccost}
c\leq\max_{y}\left[H_{0}(X_{A'})_{\rho^{y}}-H_{\min}(X_{A'}|R)_{\rho^{y}}\right]+4\cdot\log\frac{1}{\eps}+1+\log\left\lceil2\cdot\log|X_{A'}|\right\rceil
\end{align}
bits to Bob (and we now rename $X_{A'_{2}}$ to $X_{B_{2}}$ and $Y_{A}$ to $Y_{B}$). Then Bob applies
\begin{align}
S_{B_{1}Y_{B}}=\sum_{y=1}^{Q}S_{B_{1}^{y}}\otimes\proj{y}_{Y_{B}}\ ,
\end{align}
followed by the isometry
\begin{align}\label{eq:decoding}
U_{X_{B_{1}}X_{B_{2}}Y_{B}\rightarrow X_{B}Y_{B}}=\sum_{y=0}^{Q}U_{X_{B_{1}^{y}}X_{B_{2}^{y}}\rightarrow X_{B}}\otimes\proj{y}_{Y_{B}}\ .
\end{align}
We obtain a sub-normalized state
\begin{align}
\sigma_{AX_{A}X_{B}RY_{B}}=\sum_{y=0}^{Q}q_{y}\cdot\tilde{\rho}^{y}_{AX_{A}X_{B}R}\otimes\proj{y}_{Y_{B}}\ ,
\end{align}
with $\tilde{\rho}^{y}_{AX_{A}X_{B}R}\approx_{\eps}\cI_{X_{A'}\rightarrow X_{B}}(\rho^{y}_{AX_{A}X_{A'}R})$ for $y=0,1,\ldots,Q$. By the (quasi) convexity of the purified distance in its arguments (Lemma~\ref{lem:conv}), and the monotonicity of the purified distance under partial trace, we have
\begin{align}
\sigma_{AX_{A}X_{B}R}\approx_{\eps}\cI_{X_{A'}\rightarrow X_{B}}(\bar{\rho}_{AX_{A}X_{A'}R})\ .
\end{align}
Hence, we have shown the existence of an $\eps$-error classical state splitting protocol for $\bar{\rho}_{AX_{A}X_{A'}R}$ with classical communication cost as in~\eqref{eq:ccost}. But by the monotonicity of the purified distance, and the triangle inequality for the purified distance, this implies the existence of an $\left(\eps+|X_{A'}|^{-1/2}\right)$-error classical state splitting protocol for $\rho_{AX_{A}X_{A'}R}$, with the same classical communication cost as in~\eqref{eq:ccost}.

We now proceed by simplifying~\eqref{eq:ccost}. We have $H_{0}(X_{A'})_{\rho^{y}}\leq H_{\min}(X_{A'})_{\rho^{y}}+1$ for $y=0,1,\ldots,Q$ as can be seen as follows:
\begin{align}
2^{-(y+1)}\leq\lambda_{\min}(q_y \cdot \rho_{X_{A'}}^{y})\leq\mathrm{rank}^{-1}\left(q_y \cdot  \rho_{X_{A'}}^{y}\right)\leq\left\|q_y \cdot  \rho_{X_{A'}}^{y}\right\|_{\infty}\leq2^{-y}\ ,
\end{align}
where $\lambda_{\min}(\rho_{X_{A'}}^{y})$ denotes the smallest non-zero eigenvalue of $\rho_{X_{A'}}^{y}$. Thus,
\begin{align}\label{eq:minmaxequiv}
\mathrm{rank}\left(q_y \cdot  \rho_{X_{A'}}^{y}\right)\leq2^{y+1}=2^{y}\cdot2\leq\left\|q_y \cdot  \rho_{X_{A'}}^{y}\right\|_{\infty}^{-1}\cdot2\ ,
\end{align}
and this is equivalent to the claim. Hence, we get an $(\eps+|X_{A'}|^{-1/2})$-error classical state splitting protocol for $\rho_{AX_{A}X_{A'}R}$ with classical communication cost
\begin{align}
c&\leq\max_{y}\left[H_{\min}(X_{A'})_{\rho^{y}}-H_{\min}(X_{A'}|R)_{\rho^{y}}\right]+4\cdot\log\frac{1}{\eps}+2+\log\left\lceil2\cdot\log|X_{A'}|\right\rceil\\
&\leq\max_{y}\left[H_{\min}(X_{A'})_{\rho^{y}}-H_{\min}(X_{A'}|R)_{\rho^{y}}\right]+4\cdot\log\frac{1}{\eps}+4+\log\log|X_{A'}|\ .
\end{align}
Using a lower bound for the max-information in terms of min-entropies (Lemma~\ref{lem:minbound}), and the behaviour of the max-information under projective measurements (Lemma~\ref{lem:projective}) this simplifies to
\begin{align}
c&\leq\max_{y}I_{\max}(X_{A'}:R)_{\rho^{y}}+4\cdot\log\frac{1}{\eps}+4+\log\log|X_{A'}|
\label{eq:ccost1-max}\\
&\leq I_{\max}(X_{A'}:R)_{\rho}+4\cdot\log\frac{1}{\eps}+4+\log\log|X_{A'}|\ .\label{eq:ccost2}
\end{align}
Furthermore, it easily seen from~\eqref{eq:sumcost} that
\begin{align}\label{eq:sumcost2}
c+s&\leq H_{0}(X_{A'})_{\rho}+2+\log\log|X_{A'}|\ .
\end{align}

As the last step, we reduce the classical communication and shared randomness cost by smoothing the max-information and the zero-R\'enyi entropy in~\eqref{eq:ccost2} and~\eqref{eq:sumcost2}, respectively. For that, we do not apply the protocol as described above to the state $\rho_{AX_{A}X_{A'}R}$, but pretend that we have another classically coherent (sub-normalised) state $\bar{\rho}_{AX_{A}X_{A'}R}$ that is $(\sqrt{8\eps'}+\eps')$-close to $\rho_{AX_{A}X_{A'}R}$, and then apply the protocol for $\bar{\rho}_{AX_{A}X_{A'}R}$. By the monotonicity of the purified distance, the additional error term from this is upper bounded by $\sqrt{8\eps'}+\eps'$, and by the triangle inequality for the purified distance this results in a total accuracy of $\eps+\eps'+\sqrt{8\eps'}+|X_{A'}|^{-1/2}$. We now proceed by defining $\bar{\rho}_{AX_{A}X_{A'}R}$.  Let $\tilde{\rho}_{X_{A'}R}\in\cB^{\eps'}(\rho_{X_{A'}R})$ such that
\begin{align}\label{eq:finalsmooth}
I_{\max}^{\eps'}(X_{A'}:R)_{\rho}=I_{\max}(X_{A'}:R)_{\tilde{\rho}}\ .
\end{align}
Furthermore, since the zero-R\'enyi entropy can be smoothed by applying a projection (Lemma~\ref{lem:simsmoothing}), there exists $\Pi_{X_{A'}}\in\cP(X_{A'})$ with $\Pi_{X_{A'}}\leq\id_{X_{A'}}$ such that
\begin{align}\label{eq:simsmoothing}
H_{0}^{2\eps'}(X_{A'})_{\tilde{\rho}}\geq H_{0}(X_{A'})_{\bar{\rho}}\ ,
\end{align}
%\markwilde{I still don't understand why we are using $2\eps'$ and $\sqrt{8 \eps'}$ in the above and below rather than what is there in the lemma statement ...}
with $\bar{\rho}_{X_{A'}}=\Pi_{X_{A'}}\tilde{\rho}_{X_{A'}}\Pi_{X_{A'}}\in\cB^{\sqrt{8\eps'}}(\tilde{\rho}_{X_{A'}})$ classical with respect to the basis $\{\ket{x}\}$. By the properties of the purified distance~\cite[Chapter 3]{tomamichel:thesis}, there exists a purification $\bar{\rho}_{AX_{A}X_{A'}R}\in\cB^{\sqrt{8\eps'}+\eps'}(\rho_{AX_{A}X_{A'}R})$ that is classically coherent on $X_{A}X_{A'}$ with respect to the basis $\{\ket{x}\}$. Applying the protocol for this state $\bar{\rho}_{AX_{A}X_{A'}R}$, the classical communication cost~\eqref{eq:ccost2} becomes by the monotonicity of the max-information (Lemma~\ref{lem:nonnegative}) and~\eqref{eq:finalsmooth},
\begin{align}
c\leq I_{\max}^{\eps'}(X_{A'}:R)_{\rho}+4\cdot\log\frac{1}{\eps}+4+\log\log|X_{A'}|\ ,
\end{align}
and by~\eqref{eq:simsmoothing} the sum cost~\eqref{eq:sumcost2} becomes
\begin{align}
c+s\leq H_{0}^{\eps'}(X_{A'})_{\rho}+2+\log\log|X_{A'}|\ .
\end{align}
\end{proof}

For completeness we also state a converse for the
classical communication cost of classical state splitting
of classically coherent states.

\begin{theorem}\label{thm:converse}
Let $\eps\geq0$, $\eps'>0$, and $\rho_{AX_{A}X_{A'}R}\in\cV_{\leq}(\cH_{AX_{A}X_{A'}R})$ be classically coherent on $X_{A}X_{A'}$ with respect to the basis $\{\ket{x}\}_{x\in X_{A}X_{A'}}$. Then the classical communication cost for any $\eps$-error classical state splitting protocol for $\rho_{AX_{A}X_{A'}R}$ is lower bounded by\footnote{We do not mention the cost of the shared randomness
resource, since the statement holds independently of it.}
\begin{align}
c\geq I^{\eps+\eps'}_{\max}(X_{A'}:R)_{\rho}
-\log\left(\frac{8}{(\eps')^{2}}+2\right)\ .
\end{align}
\end{theorem}

\begin{proof}
We have a look at the correlations between Bob and the reference by analyzing the max-information that Bob has about the reference (recall that this will
be a max-information of the form $I_{\max}(R:B)$ where $R$ is the reference system and $B$ here is a general label for whatever Bob's system is). At the beginning of any protocol, there is no register at Bob's side correlated with the reference and therefore the max-information that Bob has about the reference is zero. Since back communication is not allowed, we can assume that the protocol for state splitting has the following form: applying local operations at Alice's side, sending bits from Alice to Bob and then applying local operations at Bob's side. Local operations at Alice's side have no influence on the max-information that Bob has about the reference. By sending $c$ bits from Alice to Bob, the max-information that Bob has about the reference can increase, but at most by $c$ (Corollary~\ref{lem:cdimension}). By applying local operations at Bob's side, the max-information that Bob has about the reference can only decrease (Lemma~\ref{lem:mono}). So the max-information that Bob has about the reference is upper bounded by $c$. Therefore, any state $\omega_{X_{B}R}$ at the end of a state splitting protocol must satisfy $I_{\max}(R:X_{B})_{\omega}\leq c$. But we also need $\omega_{X_{B}R}\approx_{\eps}\rho_{X_{B}R}\equiv\cI_{X_{A'}\rightarrow X_{B}}(\rho_{X_{A'}R})$ by the definition of $\eps$-error state splitting (Definition~\ref{def:splitting}). Using the definition of the smooth max-information, and that the smooth max-information is approximately symmetric in its arguments (Lemma~\ref{lem:ciganovic}), we obtain the bound in the statement of the theorem.
\end{proof}

\section{Universal Measurement Compression} \label{sec:main-result}

In this section, we establish our main result: feedback and non-feedback universal measurement compression. Theorem~\ref{thm:main} characterizes the trade-off between shared randomness and classical communication required to simulate many instances of a measurement on an arbitrary input state in such a way that both the sender and receiver obtain the outcomes of the measurement (feedback simulation), and Theorem~\ref{thm:measnonmain} characterizes the trade-off for the non-feedback case when only the receiver is required to get the outcomes of the measurement.

\begin{definition}[One-shot Measurement Compression]\label{def:simulation}
Consider a bipartite system with parties Alice and Bob. Let $\delta\geq0$, and $\cM:\cL(\cH_{A})\rightarrow\cL(\cH_{X})$ be a quantum-classical channel, with quantum input $A$ and classical output $X$. A quantum protocol $\cP$ is a one-shot feedback measurement compression for $\cM$ with error $\delta$ if it consists of using $s$ bits of shared randomness, applying local operations at Alice's side, sending $c$ classical bits from Alice to Bob, applying local operations at Bob's side, and
\begin{align}
\|\cP-\overline{\Delta} \circ \cM\|_{\Diamond}\leq\delta\ ,
\end{align}
where $\overline{\Delta}:\cL(\cH_{X})\rightarrow\cL(\cH_{X_A})\otimes\cL(\cH_{X_B})$ is a classical copying map,
\begin{align}
\overline{\Delta}(\sigma) \equiv \sum_x \bra{x} \sigma \ket{x} \,\, \ket{x}\bra{x}_{X_A} \otimes \ket{x}\bra{x}_{X_B},
\end{align}
ensuring that both Alice and Bob obtain the measurement outcome. The quantity $c$ is called the classical communication cost, and $s$ is the shared randomness cost.
For the case of a non-feedback measurement compression, we only require
the following condition to hold
\begin{align}
\|\cP-\cM\|_{\Diamond}\leq\delta\ ,
\end{align}
because Alice does not need to recover the output of the simulation
in this case.\footnote{If we state the task of measurement compression as being that a verifier who is given the reference system and classical output should not be able to distinguish the true channel from the simulation, then we should also demand that the common randomness and classical communication be private from the verifier.}
\end{definition}

\begin{definition}[Universal Measurement Compression]\label{def:main}
Let $\cM:\cL(\cH_{A})\rightarrow\cL(\cH_{X})$ be a quantum-classical channel. An asymptotic  measurement compression for $\cM$ is a sequence of one-shot measurement compressions $\cP^{n}$ for $\cM^{\otimes n}$ with error $\delta_{n}$, such that $\lim_{n\rightarrow\infty}\delta_{n}=0$. The classical communication rate is $\limsup_{n\rightarrow\infty}\frac{\log c_{n}}{n}$ and the shared randomness rate is $\limsup_{n\rightarrow\infty}\frac{\log s_{n}}{n}$ (where $c_{n}$ and $s_{n}$ denote the corresponding costs for the one-shot measurement compressions).
\end{definition}

%%%%%%%%%%%%%%%%%%%%%%%%%%%%%%%%%%%%%%%%%%%%%

\subsection{Feedback Simulation}

\begin{theorem}\label{thm:main}
Let $\cM:\cL(A)\rightarrow\cL(X)$ be a quantum to classical channel. Then there exist asymptotic feedback measurement compressions for $\cM$ if and only if the classical communication rate $C$ and shared randomness rate $S$ lie in the following rate region:\footnote{Note that the two maxima in~\eqref{eq:maina} and~\eqref{eq:mainb} can be achieved for different states.}
\begin{align}
C & \geq \max_{\rho}I(X:R)_{(\cM\otimes\cI)(\rho)}\label{eq:maina}\\
C + S & \geq \max_{\rho}H(X)_{\cM(\rho)}\label{eq:mainb}\ ,
\end{align}
where $\rho_{AR}\in\cV(AR)$ is a purification of the input state $\rho_{A}\in\cS(A)$. Or equivalently, for a given shared randomness rate $S$, the optimal rate of classical communication is equal to
\begin{align}\label{eq:mainc}
C(S)=\max\left\{\max_{\rho} I(X:R)_{(\cM\otimes\cI)(\rho)},\,\,\max_{\rho} H(X)_{\cM(\rho)}-S\right\}\ .
\end{align}
In particular, when sufficient shared randomness is available, the rate of classical communication is given by
\begin{align}
C(\infty)= \max_{\rho} I(X:R)_{(\cM\otimes\cI)(\rho)}.
\end{align}
\end{theorem}

\begin{proof}
We first show that the right-hand side of~\eqref{eq:maina} is a lower bound on the classical communication rate, and that~\eqref{eq:mainb} is a lower bound on the sum rate (Propositition~\ref{prop:converse}). Then we show that these lower bounds can be achieved (Proposition~\ref{prop:achiev}). The general rate trade-off in~\eqref{eq:maina}-\eqref{eq:mainb} and \eqref{eq:mainc} immediately follows, since the shared randomness can
always be  created by classical communication.
\end{proof}

\begin{proposition}[Converse]\label{prop:converse}
Let $\cM:\cL(A)\rightarrow\cL(X)$ be a quantum to classical channel. Then we have for any asymptotic measurement compression for $\cM$ that
\begin{align}
C & \geq \max_{\rho}I(X:R)_{(\cM\otimes\cI)(\rho)}\\
C + S & \geq \max_{\rho}H(X)_{\cM(\rho)}\ ,
\end{align}
where $\rho_{AR}\in\cV(AR)$ is a purification of the input state $\rho_{A}\in\cS(A)$. 
\end{proposition}

\begin{proof}
This proposition follows from the converse for the case of a fixed IID source~\cite[Theorem 8]{Winter04}, since the asymptotic measurement compressions must in particular work for any fixed IID input state $\rho^{\otimes n}_{A}$ (for $n\rightarrow\infty$). To see this explicitly worked out with the feedback assumption, see Section~2.4
of~\cite{Wilde12}.
\end{proof}

\begin{proposition}[Achievability]\label{prop:achiev}
Let $\cM:\cL(A)\rightarrow\cL(X)$ be a quantum to classical channel. Then there exist asymptotic feedback measurement compressions for $\cM$ with\begin{align}
C & \leq \max_{\rho}I(X:R)_{(\cM\otimes\cI)(\rho)}\label{eq:achiev}\\
C + S & \leq \max_{\rho}H(X)_{\cM(\rho)}\ \label{eq:achiev2},
\end{align}
where $\rho_{AR}\in\cV(AR)$ is a purification of the input state $\rho_{A}\in\cS(A)$. 
\end{proposition}

\begin{proof}
We show the existence of a sequence of one-shot feedback measurement compressions $\cP^{n}$ for $\cM^{\otimes n}$ with asymptotically vanishing error $\eps_{n}$, a classical communication rate $\frac{c_{n}}{n}$ as in~\eqref{eq:achiev}, and a shared randomness rate $\frac{s_{n}}{n}$ such that the sum rate becomes as in~\eqref{eq:achiev2}. Without loss of generality, we choose $\cP^{n}$ to be permutation covariant.\footnote{By the following argument, every protocol can be made permutation covariant. To start with, Alice applies a random permutation $\pi$ on the input system chosen according to some shared randomness. This is then followed by the original protocol (which might not yet be permutation covariant), and Bob who undoes the permutation by applying $\pi^{-1}$ on the output system. The shared randomness cost of this procedure can be kept sub-linear in $n$ by using randomness recycling as discussed in~\cite[Section IV. D]{Bennett06}.} The post-selection technique for quantum channels (Proposition~\ref{prop:selection}) then applies and upper bounds the error by
\begin{align}\label{eq:mainstep}
\delta_{n}=\|\cM_{A\rightarrow X_{B}}^{\otimes n}-\cP^{n}_{A^n\rightarrow X^n_{B}}\|_{\Diamond}\leq(n+1)^{|A|^{2}-1}\cdot\|((\cM_{A\rightarrow X_{B}}^{\otimes n}-\cP^{n}_{A^n\rightarrow X^n_{B}})\otimes\cI_{R}^{\otimes n}\otimes\cI_{R'}(\zeta^{n}_{ARR'})\|_{1}\ ,
\end{align}
where $\zeta^{n}_{ARR'}$ is a purification of the de Finetti state $\zeta^{n}_{AR}=\int\psi_{AR}^{\otimes n}\,\, d(\psi_{AR})$ with $\psi_{AR}\in\cV(AR)$, $A\cong R$ and $d(\cdot)$ the measure on the normalized pure states on $AR$ induced by the Haar measure on the unitary group acting on $AR$, normalized to $\int d(\cdot)=1$. Hence, it is sufficient to consider simulating the measurement on a purification of the de Finetti state:
\begin{align}\label{eq:structure}
\omega_{X_{B}RR'}^{n}=\left(\cM_{A\rightarrow X_{B}}^{\otimes n}\otimes\cI_{R}^{\otimes n}\otimes\cI_{R'}\right)\left(\zeta^{n}_{ARR'}\right)\ ,
\end{align}
up to an error $o\left((n+1)^{1-|A|^{2}}\right)$ in trace distance, for an asymptotic classical communication cost smaller than~\eqref{eq:achiev}. For this, we consider a local Stinespring dilation $U_{A\rightarrow EX_{A}X_{A'}}$ of the measurement $\cM_{A\rightarrow X_{A'}}$ at Alice's side, followed by classical state splitting of
the resulting classically coherent state (Theorem~\ref{thm:splitting}). Let $U_{A^n\rightarrow E^n X^n_{A} X^n_{A'}} = U^{\otimes n}_{A\rightarrow EX_{A}X_{A'}}$ and
\begin{align}\label{eq:splitstate}
\omega_{E^n X^n_{A}X^n_{A'}R^nR'}=U_{A^n \rightarrow E^n X^n_{A} X^n_{A'}}(\zeta^{n}_{ARR'})\ .
\end{align}
As mentioned above, this map can be made permutation invariant. For fixed $\eps_{n}>0$, Theorem~\ref{thm:splitting} then assures that the map outputs a state which is
\begin{align}\label{eq:distance}
4\cdot\eps_{n}+4\sqrt{2\eps_{n}}+2\cdot|X_{A'}|^{-n/2}
\end{align}
close to~\eqref{eq:structure} in trace distance,\footnote{The trace distance is upper bounded by two times the purified distance (Lemma~\ref{lem:distance}).} for a classical communication cost
\begin{align}\label{eq:deltacost}
c_{n}\leq I_{\max}^{\eps_{n}}(X_{A'}:RR')_{\omega}+4\cdot\log\frac{1}{\eps_{n}}+4+\log\log|X_{A'}|+\log n\ ,
\end{align}
and a sum cost
\begin{align}\label{eq:deltacost2}
c_{n}+s_{n}\leq H_{0}^{\eps_{n}}(X_{A'})_{\omega}+2+\log\log|X_{A'}|+\log n\ ,
\end{align}
where the last two terms on the right in each of the above expressions come from the fact that $\log\log |X_{A'}|^n = \log\log|X_{A'}|+\log n$.
We now analyse the asymptotic behaviour of~\eqref{eq:deltacost} and~\eqref{eq:deltacost2}. By a dimension upper bound for the smooth max-information (Lemma~\ref{lem:addsub}), and the fact that we can assume $|R'|\leq(n+1)^{|A|^{2}-1}$ (Proposition~\ref{prop:selection}), we get
\begin{align}\label{eq:beforecara}
c_{n}\leq I_{\max}^{\eps_{n}}(X_{A'}:R)_{\omega}+2\cdot\log\left[(n+1)^{|A|^{2}-1}\right]+4\cdot\log\frac{1}{\eps_{n}}+4+\log\log|X_{A'}|+\log n\ .
\end{align}
By a corollary of Carath\'eodory's theorem (Lemma~\ref{lem:cara}), we write
\begin{align}\label{eq:caramain}
\zeta_{AR}^{n}=\sum_{i\in I}p_{i}(\sigma^{i}_{AR})^{\otimes n}\ ,
\end{align}
where $\sigma^{i}_{AR}\in\cV(AR)$, $I=\{1,2,\ldots,(n+1)^{2|A||R|-2}\}$, and $\{p_{i}\}_{i\in I}$ a probability distribution. Using a quasi-convexity property of the smooth max-information (Lemma~\ref{lem:quasicon}), and for
\begin{align}
\chi=2\cdot\log\left[(n+1)^{|A|^{2}-1}\right]+4\cdot\log\frac{1}{\eps_{n}}+4+\log\log|X_{A'}|+\log n\ ,
\end{align}
we obtain
\begin{align}
c_{n} &\leq I_{\max}^{\eps_{n}}(X_{A'}:R)_{(\cM^{\otimes n}\otimes\cI)(\sum_{i}p_{i}(\sigma^{i})^{\otimes n})}+\chi\\
&\leq\max_{i}I_{\max}^{\eps_{n}}(X_{A'}:R)_{\left[(\cM\otimes\cI)(\sigma^{i})\right]^{\otimes n}}+\log\left[(n+1)^{2|A||R|-2}\right]+\chi\\
&\leq\max_{\rho}I_{\max}^{\eps_{n}}(X_{A'}:R)_{\left[(\cM\otimes\cI)(\rho)\right]^{\otimes n}}+\log\left[(n+1)^{2|A||R|-2}\right]+\chi\ ,
\end{align}
where the last maximum ranges over all $\rho_{AR}\in\cV(AR)$. From the asymptotic equipartition property for the smooth max-information (Lemma~\ref{lem:aep}) we obtain
\begin{align}\label{eq:beforefinal}
c_{n}\leq n\cdot\max_{\rho}I(X_{A'}:R)_{(\cM\otimes\cI)(\rho)}+\sqrt{n}\cdot\xi(\eps_{n})-2\cdot\log\frac{\eps_{n}^{2}}{24}+\log\left[(n+1)^{2|A||R|-2}\right]+\chi\ ,
\end{align}
where $\xi(\eps_{n})=8\sqrt{13-4\cdot\log\eps_{n}}\cdot(2+\frac{1}{2}\cdot\log|A|)$. By choosing
\begin{align}\label{eq:epsn}
\eps_{n}=(n+1)^{4(1-|A|^{2})}\ ,
\end{align}
we get an asymptotic classical communication cost of
\begin{align}
c=\limsup_{n\rightarrow\infty}\frac{c_{n}}{n}\leq\max_{\rho}I(X_{A'}:R)_{(\cM\otimes\cI)(\rho)}\ ,
\end{align}
for a vanishing asymptotic error~\eqref{eq:mainstep}, \eqref{eq:distance}, \eqref{eq:epsn}:
\begin{align}
\limsup_{n\rightarrow\infty}\delta_{n}&\leq\limsup_{n\rightarrow\infty}\left[\left(4\cdot(n+1)^{4(1-|A|^{2})}+4\sqrt{2}\cdot(n+1)^{2(1-|A|^{2})}+2\cdot|X_{A'}|^{-n/2}\right)(n+1)^{|A|^{2}-1}\right]\nonumber\\
&=0\ .
\end{align}
Furthermore, we estimate the asymptotic behaviour of the sum cost~\eqref{eq:deltacost2} by using~\eqref{eq:caramain} and a quasi-convexity property of the smooth zero-R\'enyi entropy (Lemma~\ref{lem:entropyquasicon}). For $\chi'=2+\log\log|X_{A'}|+\log n$ we get
\begin{align}
c_{n}+s_{n}
& \leq
\max_{i}H_{0}^{\eps_{n}}(X_{A'})_{\cM(\sigma^{i})^{\otimes n}}+\log\left[(n+1)^{2|A||R|-2}\right]+\chi' \\
& \leq
\max_{\rho}H_{0}^{\eps_{n}}(X_{A'})_{\cM(\rho)^{\otimes n}}+\log\left[(n+1)^{2|A||R|-2}\right]+\chi'\ ,
\end{align}
where $\rho_{A}\in\cS(A)$. By the equivalence of the smooth zero-R\'enyi entropy and the smooth max-entropy (Lemma~\ref{lem:alterequiv}), and the asymptotic equipartition property for the smooth max-entropy (Lemma~\ref{lem:aep}), we arrive at
\begin{align}
c_{n}+s_{n}&
\leq
\max_{\rho}H^{\eps_{n}/2}_{\max}(X_{A'})_{\cM(\rho)^{\otimes n}}+2\cdot\log\frac{8}{\eps_{n}^{2}}+\log\left[(n+1)^{2|A||R|-2}\right]+\chi'\\
&\leq n\cdot\max_{\rho}H(X_{A'})_{\cM(\rho)}+\sqrt{n}\cdot4\sqrt{1-2\cdot\log\frac{\eps_{n}}{2}}\cdot(2+\frac{1}{2}\cdot\log|X_{A'}|) \nonumber\\
& \,\,\,\, +2\cdot\log\frac{8}{\eps_{n}^{2}}+\log\left[(n+1)^{2|A||R|-2}\right]+\chi'\ ,
\end{align}
where $\rho_{A}\in\cS(A)$. By employing~\eqref{eq:epsn}, we get for the asymptotic limit 
\begin{align}
c+s=\limsup_{n\rightarrow\infty}\frac{1}{n}(c_{n}+s_{n})\leq\max_{\rho}H(X_{A'})_{\cM(\rho)}\ , \label{eq:total-cost-bound}
\end{align}
where $\rho_{A}\in\cS(A)$.
\end{proof}

%%%%%%%%%%%%%%%%%%%%%%%%%%%%%%%%%%%%%%%%%%%%%

\subsection{Non-Feedback Simulation}
\label{sec:non-feedback-simulation}
\begin{theorem}\label{thm:measnonmain}
Let $\cM:\cL(\cH_{A})\rightarrow\cL(\cH_{X})$ be a
quantum-to-classical channel. Then there exist asymptotic
non-feedback measurement compressions for $\cM$ if and
only if the classical communication rate~$C$ and
shared randomness rate $S$ lie in the rate region
given by the union of the following regions,
\begin{align}
C&\geq\max_{\rho}I(W:R)_{\beta}\label{eq:measnonmainC}\\
C+S&\geq\max_{\rho}I(W:XR)_{\beta}\label{eq:measnonmainCS}\ ,
\end{align}
where the state $\beta_{WXR}$ has the form
\begin{align}\label{eq:betaoptimization}
\beta_{WXR}=\sum_{w,x}q_{x|w}\cdot\proj{w}_{W}\otimes\proj{x}_{X}\otimes\trace_{A}\big[(\cN_{w}\otimes\cI)(\rho_{AR})\big]\ ,
\end{align}
$\rho_{AR}\in\cV(\cH_{AR})$ is a purification of the input state $\rho_{A}\in\cS(\cH_{A})$, and the union is with respect to all decompositions of the measurement $\cM$ in terms of internal measurements $\cN=\{\cN_w\}$ and conditional post-processing distributions $q_{x|w}$. That is,
for all states $\sigma$, it should hold that
\begin{align}\label{eq:mdecomp}
\sum_x \cM_x(\sigma) \proj{x} = \sum_{x,w} q_{x|w} \, \cN_w(\sigma) \proj{x}.
% \cN:\sum_{x,w}q_{x|w}\cdot\cN_{w}=\cM\ .
\end{align}
Or equivalently, for a given shared randomness rate $S$, the optimal rate of classical communication is equal to
\begin{align}\label{eq:tragenonf}
C(S)=\min_{\cN\, : \, \sum_{w} q_{x|w} \cN_w = \cM_x}
\max\left\{\max_{\rho} I(W:R)_{\beta},\,\,\max_{\rho} I(W:XR)_{\beta}-S\right\}\ .
\end{align}
\end{theorem}

By the data processing inequality for the mutual information, it holds
that~$I(W:R)_{\beta} \geq I(X:R)_{\cM(\rho)}$, and hence, the classical
communication cost can only increase compared to a feedback simulation
(Theorem~\ref{thm:main}). However, if the savings in common randomness
consumption are larger than the increase in classical communication cost, then
there is an advantage to performing a non-feedback simulation. It follows from
the considerations in~\cite{Martens90,Wilde12} that the rate
trade-offs~\eqref{eq:mainc} and~\eqref{eq:tragenonf} become identical if and
only if the elements of the measurement to simulate are all rank-one operators.

\begin{proof}
We see from the converse for the case of a fixed IID source~\cite[Theorem
9]{Wilde12}, that the right-hand side of~\eqref{eq:measnonmainC} is a lower
bound on the classical communication rate, and that~\eqref{eq:measnonmainCS}
is a lower bound on the sum rate. This is because the asymptotic non-feedback
measurement compression must work in particular for any fixed IID input state
$\rho_{A}^{\otimes n}$ (as $n\rightarrow\infty$).

As the next step, we show that these lower bounds can be achieved. The general
rate trade-off in~\eqref{eq:measnonmainC}-\eqref{eq:measnonmainCS}
and~\eqref{eq:tragenonf} then immediately follows, since shared randomness can
always be created by classical communication.

The idea for the achievability part is as follows. Given a particular
decomposition of the measurement $\cM = \{\cM_{x}\}$ as $\left\{ \sum
_{w}q_{x|w}\cdot\cN_{w}\right\} $ as stated above, Alice and Bob just use a
feedback measurement compression protocol (as in the proof of
Theorem~\ref{thm:main}) to simulate the measurement $\cN=\{\cN_{w}\}$. This is
followed by a local simulation of the classical map $q_{x|w}$ at no cost at
Bob's side. Finally, Alice and Bob can use randomness recycling to extract
$H_{\min}(W|RX)_{\beta}$ bits of shared randomness back \cite{Bennett06}. In
the one-shot case, this leads to a classical communication cost of $I_{\max
}(W:R)_{\beta}$, and a sum cost $I_{\max}(W:RX)_{\beta}$. For technical
reasons, we smooth the states using typical projectors (see
Appendix~\ref{app:typical} for background on typical projectors) and arrive at
the rates given in the statement of the theorem.

Let $\{q_{x|w},\,\cN_{\omega}\}$ be a fixed decomposition of $\cM$. As in the
feedback case (Theorem~\ref{thm:main}) we employ the post-selection technique
(Proposition~\ref{prop:selection}) to upper bound the error for one-shot
non-feedback compressions $\cP^{n}$ for $\cM^{\otimes n}$ by
\begin{align}
\delta_{n} &  =\Vert\cM_{A\rightarrow X_{B}}^{\otimes n}-\cP_{A\rightarrow
X_{B}}^{n}\Vert_{\Diamond}\label{eq:postmeasnonf}\\
&  \leq(n+1)^{|A|^{2}-1}\cdot\Vert((\cM_{A\rightarrow X_{B}}^{\otimes
n}-\cP_{A\rightarrow X_{B}}^{n})\otimes\cI_{R}^{\otimes n}\otimes
\cI_{R^{\prime}}(\zeta_{ARR^{\prime}}^{n})\Vert_{1}\ ,
\end{align}
where $\zeta_{ARR^{\prime}}^{n}$ is a purification of the de Finetti state
$\zeta_{AR}^{n}=\int\psi_{AR}^{\otimes n}\,\,d(\psi_{AR})$ with $\psi_{AR}%
\in\cV(\cH_{AR})$, $A\cong R$ and $d(\cdot)$ the measure on the normalized
pure states on $\cH_{AR}$ induced by the Haar measure on the unitary group
acting on $\cH_{AR}$, normalized to $\int d(\cdot)=1$. Hence, it is sufficient
to consider simulating the measurement $\cM^{\otimes n}$ on a purification of
the de Finetti state
\begin{equation}
\omega_{X_{B}RR^{\prime}}^{n}=\big(\cM_{A\rightarrow X_{B}}^{\otimes n}%
\otimes\cI_{R}^{\otimes n}\otimes\cI_{R^{\prime}}\big)\left(  \zeta
_{ARR^{\prime}}^{n}\right)  \ ,
\end{equation}
up to an error $o\left(  (n+1)^{1-|A|^{2}}\right)  $ in trace distance, for an
asymptotic simulation cost smaller than in~\eqref{eq:measnonmainC}
and~\eqref{eq:measnonmainCS}. For this, the idea is to consider a local
Stinespring dilation $V_{A\rightarrow EW_{A}W_{A^{\prime}}}$ of the
measurement $\cN_{A\rightarrow W_{A}}$ at Alice's side, followed by classical
state splitting of the resulting classically coherent state (along
Theorem~\ref{thm:splitting}). Let $V_{A\rightarrow EW_{A}W_{A^{\prime}}%
}^{n}=V_{A\rightarrow EW_{A}W_{A^{\prime}}}^{\otimes n}$ and
\begin{equation}
\label{eq:definettimeasnonf}
\omega_{EW_{A}W_{A^{\prime}}RR^{\prime}}^{n}=V_{A\rightarrow EW_{A}%
W_{A^{\prime}}}^{n}\left(  \zeta_{ARR^{\prime}}^{n}\right)  \ .
\end{equation}
However, Alice and Bob will not execute the protocol with respect to the state
$\omega_{EW_{A}W_{A^{\prime}}RR^{\prime}}^{n}$ directly, but they will do so
with respect to another pure, sub-normalized state $\bar{\gamma}%
_{EW_{A}W_{A^{\prime}}RR^{\prime}}^{n}$ that is also classically coherent on
$W_{A}W_{A^{\prime}}$ with respect to the basis $\{\ket{w}\}_{w\in W_{A}}$,
and such that
\begin{equation}
\Vert\bar{\gamma}_{EW_{A}W_{A^{\prime}}RR^{\prime}}^{n}-\omega_{EW_{A}%
W_{A^{\prime}}RR^{\prime}}^{n}\Vert_{1}\leq\eps_{n}\ ,\label{eq:gammastate}
\end{equation}
for some $\eps_{n}>0$. By a corollary of Carath\'{e}odory's theorem
(Lemma~\ref{lem:cara}), we write
\begin{equation}
\zeta_{AR}^{n}=\sum_{i\in I}p_{i}\cdot(\sigma_{AR}^{i})^{\otimes n}\ ,
\label{eq:carameasnonf}
\end{equation}
where $\sigma_{AR}^{i}\in\cV(\cH_{AR})$, $I=\{1,2,\ldots,(n+1)^{2|A||R|-2}\}$,
and $\{p_{i}\}_{i\in I}$ a probability distribution. From this, we define
\begin{equation}
\gamma_{EW_{A}W_{A^{\prime}}R}^{i,n}=[(V_{A\rightarrow EW_{A}W_{A^{\prime}}%
}\otimes\cI_{R})(\sigma_{AR}^{i})]^{\otimes n} ,
\end{equation}
as well as its reduction as a classical-quantum state $\gamma_{W_{A}R}^{i,n}$
on the systems $W_{A^{\prime}}^{n}R^{n}$:
\begin{equation}
\gamma_{W_{A}R}^{i,n}=\sum_{w^{n}}p_{W^{n}|i}(w^{n}|i)\,\proj{w^n}_{W_{A^{\prime
}}^{n}}\otimes\gamma_{R^{n}}^{i,w^{n}} ,
\end{equation}
for some distribution $p_{W^{n}|i}(w^{n}|i)$. On this state, we act with
typical projectors to flatten its spectrum as we need, defining the projected
state $\bar{\gamma}_{W_{A}R}^{i,n}$ as follows:
\begin{equation}\label{eq:typicalmeasnonf}
\bar{\gamma}_{W_{A}R}^{i,n}=\sum_{w^{n}}p_{W^{n}|i}(w^{n}|i)\,\Pi_{\delta
}^{W^{n}|i}\,\proj{w^n}_{W_{A^{\prime}}^{n}}\,\Pi_{\delta}^{W^{n}|i}\otimes
\Pi_{\gamma^{i},\delta}^{n}\,\Pi_{\gamma^{i,w^{n}},\delta}^{n}\,\gamma_{R^{n}%
}^{i,w^{n}}\,\Pi_{\gamma^{i,w^{n}},\delta}^{n}\,\Pi_{\gamma^{i},\delta}^{n},
\end{equation}
where $\Pi_{\delta}^{W^{n}|i}$ is a typical projector corresponding to the
distribution $p_{W^{n}|i}(w^{n}|i)$, $\Pi_{\gamma^{i,w^{n}},\delta}^{n}$ is a
conditionally typical projector corresponding to the conditional state
$\gamma^{i,w^{n}}$ on the system $R^{n}$, and $\Pi_{\gamma^{i},\delta}^{n}$ is
a typical projector corresponding to the state $\gamma_{R}^{i,n}$ (see
Appendix~\ref{app:typical} for details of typical projectors).
It follows from the properties
of typical projectors that the projected state $\bar{\gamma}_{W_{A}R}^{i,n}$
becomes arbitrarily close in trace distance to the original state
$\gamma_{W_{A}R}^{i,n}$:
\begin{equation}
\Vert\gamma_{W_{A}R}^{i,n}-\bar{\gamma}_{W_{A}R}^{i,n}\Vert_{1}\leq
\frac{\varepsilon_{n}^{2}}{4} ,
\end{equation}
for some $\varepsilon_{n}>0$ and sufficiently large $n$. The equivalence of the trace distance and the purified distance (Lemma~\ref{lem:distance}) together with Uhlmann's theorem then imply the existence of some subnormalized pure state $\bar{\gamma}_{EW_{A}W_{A^{\prime}}R}^{i,n}$ such that
\begin{align}
P(\gamma_{EW_{A}W_{A^{\prime}}R}^{i,n},\bar{\gamma}_{EW_{A}W_{A^{\prime}}R}^{i,n})\leq\frac{\varepsilon_{n}}{2}\ .
\end{align}
Hence, we get by~\eqref{eq:carameasnonf} and~\eqref{eq:definettimeasnonf} that%
\begin{equation}
\bar{\gamma}_{EW_{A}W_{A^{\prime}}R}^{n}=\sum_{i\in I}p_{i}\cdot\bar{\gamma
}_{EW_{A}W_{A^{\prime}}R}^{i,n}%
\end{equation}
is $\varepsilon_{n}$-close to $\omega_{EW_{A}W_{A^{\prime}}R}^{n}$ in purified distance. By features of the purified distance~\cite[Chapter 3]{tomamichel:thesis}, and the equivalence of the trace distance and the purified distance (Lemma~\ref{lem:distance}), we then get that there exists an extension $\bar{\gamma}_{EW_{A}W_{A^{\prime}}RR^{\prime}}^{n}$ of $\bar {\gamma}_{EW_{A}W_{A^{\prime}}R}^{n}$ with the desired properties such that~\eqref{eq:gammastate} holds.

Alice and Bob will now act with a classical state splitting protocol for
$W_{A}W_{A^{\prime}}$ with respect to the classically coherent state
$\bar{\gamma}_{EW_{A}W_{A^{\prime}}RR^{\prime}}^{n}$. However, we do not
directly use our result about classical state splitting
(Theorem~\ref{thm:splitting}), but instead employ a non-smooth version that is
implicit in the proof of Theorem~\ref{thm:splitting}. It follows
from~\eqref{eq:ccost1-max} and \eqref{eq:ccost2} that for an $(\eps_{n}%
+|W_{A^{\prime}}|^{-n})$-error (in purified distance) classical state
splitting protocol for $W_{A}W_{A^{\prime}}$, a classical communication cost%
\begin{equation}\label{eq:ccostmeasnonf}
c_{n}\leq\max_{y}I_{\max}(W_{A^{\prime}}:RR^{\prime})_{\bar{\gamma}^{n,y}%
}+4\cdot\log\frac{1}{\eps_{n}}+4+\log\log|W_{A^{\prime}}|+\log n
\end{equation}
is achievable, and it follows from~\eqref{eq:sumcost} and \eqref{eq:sumcost2}
that the sum cost becomes
\begin{equation}\label{eq:sumcostmeasnonf}
c_{n}+s_{n}\leq\max_{y}H_{0}(W_{A^{\prime}})_{\bar{\gamma}^{n,y}}+2+\log
\log|W_{A^{\prime}}|+\log n\ ,
\end{equation}
where the measurement outcomes $y$ are with respect to the pre-processing
measurement defined in~\eqref{eq:preproc}. This provides Bob with the
measurement outcomes of $\cN$ for the fixed de Finetti type input state
$\zeta_{ARR^{\prime}}$, and a total error of $(3\cdot\eps_{n}+2\cdot
|W_{A^{\prime}}|^{-n})$ in trace distance.
A local simulation of the classical map $q_{x^{n}|w^{n}}$ at no cost at Bob's
side then provides Bob with the measurement outcomes of $\cM$ as desired
(again for the fixed de Finetti type input state $\zeta_{ARR^{\prime}}$ and
the same error). However, the sum cost of this non-feedback measurement
simulation can be reduced by invoking an additional randomness recycling step
as in Ref.~\cite{Bennett06}.
We do this by having Alice and Bob apply, conditioned on $y$, a strong
classical min-entropy extractor on $W$ against the (quantum) side information
$XRR^{\prime}$ (Proposition~\ref{prop:extractor}), and this lowers the sum
cost to
\begin{equation}
c_{n}+s_{n}\leq\max_{y}\big(H_{0}(W_{A^{\prime}})_{\bar{\gamma}^{n,y}}%
-H_{\min}(W_{A^{\prime}}|RR'X_{A^{\prime}})_{\bar{\gamma}^{n,y}}\big)+4\cdot
\log\frac{1}{\eps_{n}}+2+\log\log|W_{A^{\prime}}|\ ,
\end{equation}
for an additional error $\eps_{n}$ in trace distance, leading to a total error
of
\begin{equation}\label{eq:totalerrormeasnonf}
(4\cdot\eps_{n}+2\cdot|W_{A^{\prime}}|^{-n})
\end{equation}
in trace distance. The min-entropy extractor is performed with respect to the
following typical projected state, in order to increase the amount of
randomness that can be extracted:%
\begin{multline}
\bar{\gamma}_{XW_{A}R}^{i,n}=\sum_{w^{n},x^{n}}q\left(  x^{n}|w^{n}\right)
p_{W^{n}|i}(w^{n}|i)\,\Pi_{\delta}^{X^{n}|W^{n},i}\ \proj{x^n}_{X^{n}}%
\ \Pi_{\delta}^{X^{n}|W^{n},i}\otimes\\
\Pi_{\delta}^{W^{n}|i}\,\proj{w^n}_{W_{A^{\prime}}^{n}}\,\Pi_{\delta}%
^{W^{n}|i}\otimes\Pi_{\gamma^{i},\delta}^{n}\,\Pi_{\gamma^{i,w^{n}},\delta
}^{n}\,\gamma_{R^{n}}^{i,w^{n}}\,\Pi_{\gamma^{i,w^{n}},\delta}^{n}%
\,\Pi_{\gamma^{i},\delta}^{n}\ .
\end{multline}

In the rest of the proof, we bring the classical communication
cost~\eqref{eq:ccostmeasnonf} and the sum cost~\eqref{eq:sumcostmeasnonf} into
the right form, and show that the asymptotic error for the measurement
simulation~\eqref{eq:postmeasnonf} becomes zero. By the behavior of the
max-information under projective measurements (Corollary~\ref{lem:projective}%
), a dimension upper bound for the max-information (Lemma~\ref{lem:addsub}),
the fact that we can assume $|R^{\prime|A|^{2}-1}$
(Proposition~\ref{prop:selection}), and a quasi-convexity property of the
max-information (Lemma~\ref{lem:quasicon}), we get
\begin{equation}
c_{n}\leq\max_{i\in I}I_{\max}(W_{A^{\prime}}:R)_{\bar{\gamma}^{i,n}}+\chi\ ,
\end{equation}
where
\begin{equation}
\chi=2\cdot\log\big((n+1)^{|A|^{2}-1}\big)+\log\big((n+1)^{2|A||R|-2}%
\big)+4\cdot\log\frac{1}{\eps_{n}}+4+\log\log|W_{A^{\prime}}|+\log n\ .
\end{equation}
By an upper bound on the max-information (Lemma~\ref{lem:minbound}), and a
lower bound on the conditional min-entropy (Lemma~\ref{lem:minlower}), this
can be estimated to be
\begin{equation}
c_{n}\leq\max_{i\in I}\big(H_{R}(W_{A^{\prime}})_{\bar{\gamma}^{i,n}}-H_{\min
}(W_{A^{\prime}}R)_{\bar{\gamma}^{i,n}}+H_{0}(R)_{\bar{\gamma}^{i,n}%
}\big)+\chi\ .
\end{equation}
By~\eqref{eq:typicalmeasnonf}, as well as the properties of typical
projectors (see Appendix~\ref{app:typical}), we get
\begin{align}
c_{n} &  \leq n\cdot\max_{i\in I}\big(H(W_{A^{\prime}})_{\gamma^{i}%
}-H(W_{A^{\prime}}R)_{\gamma^{i}}+H(R)_{\gamma^{i}}\big)+5nc\delta+\chi\\
&  \leq n\cdot\max_{\rho}\big(H(W_{A^{\prime}})_{\cN(\rho)}-H(W_{A^{\prime}%
}R)_{(\cN\otimes\cI)(\rho)}+H(R)_{\rho}\big)+5nc\delta+\chi\ ,
\end{align}
where $\rho_{AR}\in\cV(\cH_{AR})$, $c$ is a constant, and $\delta > 0$ is
the typicality tolerance.

By choosing
\begin{equation}
\eps_{n}=(n+1)^{4(1-|A|^{2})}\ ,
\end{equation}
we finally get an asymptotic classical communication cost of
\begin{equation}
c=\limsup_{n\rightarrow\infty}\frac{c_{n}}{n}\leq\max_{\rho}I(W_{A^{\prime}%
}:R)_{\beta}\ ,
\end{equation}
where $\rho_{AR}\in\cV(\cH_{AR})$, $\beta_{W_{A^{\prime}}R}$ is as
in~\eqref{eq:betaoptimization}, and a vanishing asymptotic
error~\eqref{eq:postmeasnonf}, \eqref{eq:totalerrormeasnonf},
\begin{equation}
\limsup_{n\rightarrow\infty}\delta_{n}\leq\limsup_{n\rightarrow\infty
}\Big((n+1)^{|A|^{2}-1}\cdot\big(4\cdot(n+1)^{4(1-|A|^{2})}+2\cdot
|X_{A^{\prime}}|^{-n}\big)\Big)=0\ .
\end{equation}
For the sum cost~\eqref{eq:sumcostmeasnonf} we get by the definition of the
measurement in~\eqref{eq:preproc} with outcomes $y$, and a line of argument as
in~\eqref{eq:minmaxequiv} that
\begin{align}
c_{n}+s_{n} &  \leq\max_{y}\big(H_{\min}(W_{A^{\prime}})_{\bar{\gamma}^{n,y}%
}-H_{\min}(W_{A^{\prime}}|RR'X_{A^{\prime}})_{\bar{\gamma}^{n,y}}%
\big)\nonumber\\
&  +4\cdot\log\frac{1}{\eps_{n}}+2+\log\log|W_{A^{\prime}}|+\log n\\
&  \leq\max_{y}I_{\max}(W_{A^{\prime}}:RX_{A^{\prime}})_{\bar{\gamma}^{n,y}%
}\nonumber\\
&  +2\cdot\log\big((n+1)^{|A|^{2}-1}\big)+4\cdot\log\frac{1}{\eps_{n}}+2+\log\log|W_{A^{\prime}}|+\log n\ ,
\end{align}
where we used a lower bound on the max-information (Lemma \ref{lem:minbound}), as well as a dimension upper bound for the max-information (Lemma~\ref{lem:addsub}), and the fact that $|R'|\leq(n+1)^{|A|^{2}-1}$ (Proposition~\ref{prop:selection}). Using similar arguments
(see Appendix~\ref{app:typical}) as in the estimation of the
classical communication cost, we arrive at
\begin{equation}
c+s=\limsup_{n\rightarrow\infty}\frac{c_{n}+s_{n}}{n}\leq\max_{\rho
}I(W_{A^{\prime}}:RX_{A^{\prime}})_{\beta}\ ,
\end{equation}
where $\rho_{AR}\in\cV(\cH_{AR})$, and $\beta_{W_{A^{\prime}}RX_{A^{\prime}}}$
is as in~\eqref{eq:betaoptimization}. By minimizing over all decompositions of
the measurement $\cM$ as in~\eqref{eq:mdecomp}, the claim follows.
\end{proof}

%%%%%%%%%%%%%%%%%%%%%%%%%%%%%%%%%%%%%%%%%%%%%

\section{Extensions and Applications}

\label{sec:applications}

\paragraph{Structured State Splitting Scheme} The state splitting protocol presented in Theorem~\ref{thm:splitting} has the drawback that the permutations $U_y$ used by Bob must be chosen at random and little is known about the structure of the unitaries $V_y$. We can remedy this by basing the state splitting protocol used in Theorem~\ref{thm:splitting} on a modified state merging protocol instead of that in Lemma~\ref{lem:merging}.  The new protocol has the advantage that Alice's classical operation $P$ (recall that the roles are reversed) is a \emph{linear} function rather than an arbitrary permutation, though still randomly-chosen, and Bob's unitary operation $V$ is based on the decoder of an information reconciliation protocol. We now give a sketch of this modified state merging protocol.  

The protocol is based on the observation from~\cite{boileaurenes,reneshabi} that state merging is a by-product of an entanglement distillation protocol in which Alice measures the stabilizers of a Calderbank-Shor-Steane (CSS) code such that, given the resulting (classical) syndrome results, Bob could determine both the amplitude (logical $X$ value) and phase (logical $Z$ value)\footnote{Associating $X$ with amplitude instead of phase contravenes the usual convention in the QECC literature, but better fits the notation of the current paper.} of Alice's remaining encoded system by using his systems. Indeed, for state merging of classically coherent states such as $\rho_{X_AX_BBR}$ in Lemma~\ref{lem:merging}, the situation is considerably simpler since Bob can already determine the amplitude of Alice's system $X_A$ by measuring $X_B$. For simplicity, let us regard $X_A$ as a collection of $k=\log|X_A|$ qubits.

Thus, from the analysis of~\cite{boileaurenes,reneshabi}, all that remains is for Alice to measure a sufficient number of phase stabilizers from an error-correcting code to enable Bob to determine the phase of her encoded systems by using the syndromes and his systems $X_B$ and $B$, with probability of error at most $\epsilon$. Use of a linear code ensures that Alice does not damage Bob's amplitude information in the course of trying to increase his phase information.  Since the task at hand is equivalent to information reconciliation, the number of phase stabilizers needed for this purpose is no more than $H_{\rm max}(\widetilde X_A|X_BB)_\rho+2\log\frac1\epsilon+4$~\cite{renes10}, where $\widetilde{X}_A$ denotes the phase observable conjugate to the amplitude observable $X_A$. 

To measure the phase stabilizers, Alice can apply a suitable unitary operation to all of her systems and then simply measure the phases of a certain subset of the outputs which correspond to the stabilizers~\cite{NieChu00Book}. But for stablizer codes, this unitary just implements a linear transformation in the phase basis of the $k$ qubits, which can equally well be regarded as a linear transformation in the amplitude basis $\{\ket{x}\}_{x\in X_A}$. Therefore, just as in the original protocol, Alice applies a ``classical'' transformation of her system and sends one part of the output to Bob. 

For his part, Bob can complete the state merging protocol by coherently implementing the decoder from the information reconciliation protocol, a construction of which based on the pretty good measurement is given in~\cite{renes10}. 

Finally, the number of entangled systems generated in the state merging protocol is equal to the number of systems left at Alice's side, or $|X_A|-H_{\rm max}(\widetilde X_A|X_BB)_\rho-2\log\frac1\epsilon-4$. Lemma~\ref{lem:minmax} shows that this is in fact greater than $H_{\rm min}(X_A|R)_\rho-2\log\frac1\epsilon-4$. Thus, the stabilizer-based state merging protocol achieves the same costs as the state merging protocol of Lemma~\ref{lem:merging} (up to terms of order $\log\frac1\epsilon$).

\paragraph{Fixed IID Source}
The case of a fixed IID source also follows easily from our analysis. We can simply apply
the one-shot protocol from Theorem~\ref{thm:splitting} to the case of a fixed IID source and then invoke the
asymptotic equipartition property for the smooth max-information and the smooth max-entropy.
In this way, we provide an alternative proof of this special case that avoids the use of typical projectors 
and the operator Chernoff bound~\cite{Winter04}.

\paragraph{Instrument Compression}
In Winter's original paper on measurement compression, additional arguments were required to establish that a POVM (positive operator-valued measure) compression protocol
can function as an instrument compression protocol, where for an instrument compression protocol,
Alice and Bob receive the classical outcomes of the measurement while Alice obtains the
post-measurement states (see Section~V of \cite{Winter04}). We note that our protocol
here already functions as an instrument compression protocol due to our use of the classical state splitting
protocol as a coding primitive.

\paragraph{Universal Measurement Compression with Quantum Side Information}
We briefly mention that there is no point in considering a protocol for universal measurement compression
with quantum side information (similar to the observation in Section~6.3 of \cite{DWHW12}).
In such a scenario, the receiver would obtain some quantum side information
correlated with the state on which the measurement should be simulated (see \cite{Wilde12} for the case of
measurement compression with quantum side information for a fixed IID source). Though, since a universal protocol
should simulate the measurement with respect to an \emph{arbitrary} input state, a special case of this
input is one in which
the quantum side information and input state are in a product state. Thus, the universal protocol given here
is suitable for this case. This occurs simply because our simulation is with respect to the diamond norm, and
the diamond norm is known to be robust under tensoring with other systems upon which the channel
of interest does not act.

Another way to see this is that one could imagine devising a protocol for which quantum side information
is taken into account. Based on the results in Ref.~\cite{Wilde12}, we would expect the rate of
classical communication for such a protocol to be equal to the following information quantity:
$% \begin{equation}
\max_{\rho}I(X:R\,|\,B)_{(\cI\otimes\cM\otimes\cI)(\rho)}\ 
$% \end{equation}
, where $\rho_{AB}\in\cS(AB)$ is an input state with quantum side information in the
system $B$, and $\rho_{RAB}\in\cV(RAB)$ is a purification of $\rho_{AB}$.
Though, as shown in Theorem~16 of \cite{DWHW12}, the above information quantity is actually equal to the information quantity in \eqref{eq:max-info-gain}, so that there is no improvement in the communication
rate from the availability of quantum side information.

%%%%%%%%%%%%%%%%%%%%%%%%%%%%%%%%%%%%%%%%%%%%%

\section{Conclusion}

\label{sec:conclusion}

We have justified the information-theoretic measure
in \eqref{eq:max-info-gain} as quantifying the
information gain of a quantum measurement, by providing
an operational interpretation in terms
of a protocol for universal measurement compression.
The main tools used to prove this result
are the post-selection technique for quantum channels
and a novel classical state splitting protocol
based on permutation-based extractors.

There are a number of open questions to consider going forward from here.
Given that there are applications of ``information gain'' or ``entropy reduction'' in thermodynamics \cite{J09} and quantum feedback control \cite{DJJ01}, it would be interesting to explore whether the quantity in
\eqref{eq:max-info-gain} has some application in these domains.
Also, Buscemi {\it et al}. showed that the static measure of information gain in \eqref{eq:winter-info-gain} plays a role in quantifying the trade-off between information extraction and disturbance \cite{BHH08}, and it would be interesting to determine if there is a role in this setting for the information quantity in \eqref{eq:max-info-gain}.

\section*{Acknowledgments}

%We gratefully acknowledge useful discussions with...

We acknowledge discussions with Francesco Buscemi, Matthias Christandl,
Nilanjana Datta, Patrick Hayden, Renato Renner, and Marco Tomamichel.
MB and JMR are supported by the Swiss National Science Foundation
through the National Centre of Competence in Research \lq Quantum
Science and Technology\rq. MB is also supported by Swiss
National Science Foundation grants PP00P2-128455 and 20CH21-138799
and German Science Foundation grant CH 843/2-1. JMR is also
supported by European Research Council grant 258932. MMW
acknowledges support from the Centre de Recherches
Math\'{e}matiques at the University of Montreal.

%%%%%%%%%%%%%%%%%%%%%%%%%%%%%%%%%%%%%%%%%%%%%

\appendix

\section{Entropies}

\begin{lemma}\cite[Lemma 3.1.10]{renato:diss}\label{lem:minlower}
Let $\rho_{AB}\in\cS_{\leq}(\cH_{AB})$. Then we have that
\begin{align}
H_{\min}(A|B)_{\rho}\geq H_{\min}(AB)_{\rho}-H_{0}(B)_{\rho}\ .
\end{align}
\end{lemma}

The max-mutual information is monotone under local operations.

\begin{lemma}\cite[Lemma B.14]{Berta09_2}\label{lem:mono}
Let $\rho_{AB}\in\cS_{\leq}(AB)$, and let $\cE$ be a quantum channel of the form $\cE=\cE_{A}\otimes\cE_{B}$. Then we have that
\begin{align}
I_{\max}(A:B)_{\rho}\geq I_{\max}(A:B)_{\cE(\rho)}\ .
\end{align}
\end{lemma}

%The max-mutual information of classical-quantum states is upper bounded by the dimension of the classical system.

%\begin{lemma}\label{lem:cdimension}
%Let $\eps\geq0$, and let $\rho_{XR}\in\cS(XR)$ be classical on $X$ with respect to the basis $\{\ket{x}\}_{x\in X}$. Then we have that
%\begin{align}
%I_{\max}(X:R)_{\rho}\leq\log|X|\ .
%\end{align}
%\end{lemma}

%\begin{proof}
%By~\cite[Lemma B.2]{Berta09_2} we have
%\begin{align}
%I_{\max}(X:R)_{\rho}=H_{0}(X)_{\rho}-H_{\min}(X|R)_{\rho_{R|X}}\ ,
%\end{align}
%where $\rho_{R|X}=(\rho_{X}\otimes\id_{R})^{-1/2}\frac{\rho_{XR}}{\mathrm{rank}(\rho_{X})}(\rho_{X}\otimes\id_{R})^{-1/2}$. Since $\rho_{R|X}$ is classical on $X$ with respect to the basis $\{\ket{x}\}_{x\in X}$, $H_{\min}(X|R)_{\rho_{R|X}}\geq0$ by~\cite[Lemma 3.1.9]{renato:diss}.
%\end{proof}

The max-information can be upper and lower bounded in terms of entropies.

\begin{lemma}\cite[Lemma B.10]{Berta09_2}\label{lem:minbound}
Let $\rho_{AB}\in\cS_{\leq}(AB)$. Then we have that
\begin{align}
H_{R}(A)_{\rho}-H_{\min}(A|B)_{\rho}\geq I_{\max}(A:B)_{\rho}\geq H_{\min}(A)_{\rho}-H_{\min}(A|B)_{\rho}\ ,
\end{align}
\end{lemma}
$H_R(\rho)$ is defined as the negative logarithm of the smallest
eigenvalue of $\rho$ on its support \cite{Berta09_2}.

The following lemma is about the behavior of the max-information under projective measurements.

\begin{lemma}\cite[Corollary B.16]{Berta09_2}\label{lem:projective}
Let $\rho_{AB}\in\cS_{\leq}(AB)$, and let $P=\left\{P_{A}^{i}\right\}_{i\in I}$ be a collection of projectors that describe a projective measurement on $A$. For $\mathrm{tr}\left[P^{i}_{A}\rho_{A}\right]\neq0$, let $p_{i}=\mathrm{tr}\left[P^{i}_{A}\rho_{A}\right]$, and $\rho_{AB}^{i}=p_{i}^{-1}\cdot P^{i}_{A}\rho_{AB}P^{i}_{A}$. Then we have that
\begin{align}
I_{\max}(A:B)_{\rho}\geq\max_{i}I_{\max}(A:B)_{\rho^{i}}\ ,
\end{align}
where the maximum ranges over all $i$ for which $\rho_{AB}^{i}$ is defined.
\end{lemma}

\begin{lemma}\label{lem:cqmax}
Let $\eps\geq0$, and let $\rho_{XR}\in\cS(XR)$ be classical on $X$ with respect to the basis $\{\ket{x}\}_{x\in X}$. Then there exists $\bar{\rho}_{XR}\in\cB^{\eps}(\rho_{XR})$ classical on $X$ with respect to the basis $\{\ket{x}\}_{x\in X}$ such that
\begin{align}
I_{\max}^{\eps}(X:R)_{\rho}=I_{\max}(X:R)_{\bar{\rho}}\ .
\end{align}
\end{lemma}

\begin{proof}
This is standard and can be proven exactly as in~\cite[Proposition 5.8]{tomamichel:thesis}.
\end{proof}

We need the following monotonicity of the max-information.

\begin{lemma}\label{lem:nonnegative}
Let $\rho_{AR}\in\cS(AR)$, and $\Pi_{A}\in\cP(A)$ with $\Pi_{A}\leq\id_{A}$. Then we have that
\begin{align}
I_{\max}(A:R)_{\rho}\geq I_{\max}(A:R)_{\Pi\rho\Pi}\ .
\end{align}
\end{lemma}

\begin{proof}
Let $\sigma_{R}\in\cS(R)$, and let $\lambda\in\mathbb{R}$ be such that $I_{\max}(A:R)_{\rho}=D_{\max}(\rho_{AR}\|\rho_{A}\otimes\sigma_{R})=\log\lambda$. Then we have that $\lambda\cdot\rho_{A}\otimes\sigma_{R}\geq\rho_{AR}$, and with this
\begin{align}
\lambda\cdot\rho_{A}\otimes\sigma_{R}\geq\lambda\cdot\Pi_{A}\rho_{A}\Pi_{A}\otimes\sigma_{R}\geq\Pi_{A}\rho_{AR}\Pi_{A}\ .
\end{align}
Hence, we have $\log\lambda\geq D_{\max}(\Pi_{A}\rho_{AR}\Pi_{A}\|\Pi_{A}\rho_{A}\Pi_{A}\otimes\sigma_{R})\geq I_{\max}(A:R)_{\Pi\rho\Pi}$.
\end{proof}

The following is a bound on the increase of the smooth max-information when an additional subsystem is added.

\begin{lemma}\cite[Lemma B.9]{Berta09_2}\label{lem:addsub}
Let $\eps\geq0$, and let $\rho_{ABR}\in\cS(ABR)$. Then we have that
\begin{align}
I_{\max}^{\eps}(A:BR)_{\rho}\leq I_{\max}^{\eps}(A:B)_{\rho}+2\cdot\log|R|\ .
\end{align}
\end{lemma}

The following is a strengthening of the bound in Lemma~\ref{lem:addsub} when the additional system is classical.

\begin{lemma}\label{lem:cdimension}
Let $\rho_{ABX}\in\cS(ABX)$ be classical on $\cH_{X}$ with respect to the basis $\{\ket{x}\}_{x\in X}$. Then we have that
\begin{align}
I_{\max}(A:BX)_{\rho}\leq I_{\max}(A:B)_{\rho}+\log|X|\ .
\end{align}
\end{lemma}

\begin{proof}
Let $\sigma_{B}\in\cS(B)$ be such that
\begin{align}
I_{\max}(A:B)_{\rho}=D_{\max}\left(\rho_{AB}\|\rho_{A}\otimes\sigma_{B}\right)=\log\mu\ ,
\end{align}
that is, $\mu\in\bR$ is minimal such that $\mu\cdot\rho_{A}\otimes\sigma_{B}\geq\rho_{AB}$. This implies $\mu\cdot\rho_{A}\otimes\sigma_{B}\otimes\frac{\id_{X}}{|X|}\geq\frac{1}{|X|}\cdot\rho_{AB}\otimes\id_{X}$. But we have by~\cite[Lemma 3.1.9]{renato:diss} that $\rho_{AB}\otimes\id_{X}\geq\rho_{ABC}$, and hence $\mu\cdot\rho_{A}\otimes\sigma_{B}\otimes\frac{\id_{X}}{|X|}\geq\frac{1}{|X|}\cdot\rho_{ABX}$. Now, let $\lambda\in\bR$ be minimal such that $\lambda\cdot\rho_{A}\otimes\sigma_{B}\otimes\frac{\id_{X}}{|X|}\geq\rho_{ABX}$. Thus, it follows that $\lambda\leq\mu\cdot|X|$, and from this we get
\begin{align}
I_{\max}(A:BX)_{\rho}&\leq D_{\max}(\rho_{ABX}\|\rho_{A}\otimes\sigma_{B}\otimes\frac{\id_{X}}{|X|})=\log\lambda\\
&\leq D_{\max}(\rho_{AB}\|\rho_{A}\otimes\sigma_{B})+\log|X|=I_{\max}(A:B)_{\rho}+\log|X|\ .
\end{align}
\end{proof}

The smooth max-information is quasi-convex in its argument in the following sense.

\begin{lemma}\cite[Lemma B.18]{Berta09_2}\label{lem:quasicon}
Let $\eps\geq0$, and let $\rho_{AB}=\sum_{i\in I}p_{i}\rho_{AB}^{i}\in\cS_{\leq}(AB)$ with $\rho_{AB}^{i}\in\cS_{\leq}(AB)$ for $i\in I$. Then we have that
\begin{align}
I_{\max}^{\eps}(A:B)_{\rho}\leq\max_{i\in I}I_{\max}^{\eps}(A:B)_{\rho^{i}}+\log|I|\ .
\end{align}
\end{lemma}

The following is a quasi-convexity property of the zero-R\'enyi entropy.

\begin{lemma}\cite[Lemma 26]{entCost}\label{lem:entropyquasicon}
Let $\eps\geq0$, and let $\rho_{A}=\sum_{j=1}^{N}p_{j}\rho_{A}^{j}\in\cS(A)$ with $p_{j}>0$ for $j=1,\ldots,N$. Then we have that
\begin{align}
H_{0}^{\eps}(A)_{\rho}\leq\max_{j}H_{0}^{\eps}(A)_{\rho^{j}}+\log N\ .
\end{align}
\end{lemma}

The smooth max-entropy and smooth zero-R\'enyi entropy are equivalent in the following sense.

\begin{lemma}\label{lem:alterequiv}
Let $\eps>0$, $\eps'\geq0$, and $\rho_{A}\in\cS(A)$. Then we have that
\begin{align}
H_{0}^{\eps'}(A)_{\rho}\geq H_{\max}^{\eps'}(A)_{\rho}>H_{0}^{\eps'+\sqrt{2\eps}}(A)_{\rho}-2\cdot\log\frac{1}{\eps}\ .
\end{align}
\end{lemma}

\begin{proof}
Since the (unconditional) max-entropy is the R\'enyi entropy of order $1/2$, the first inequality just follows from the ordering of the R\'enyi entropies~\cite{renyi60,Petz_Book}.

The idea for the proof of the second inequality is from the supplementary material~\cite[Lemma 13]{Berta09}. Let $\sigma_{A}\in\cB^{\eps'}(\rho_{A})$ such that $H_{\max}^{\eps'}(A)_{\rho}=H_{\max}(A)_{\sigma}$, and let $\sigma_{A}=\sum_{i}t_{i}\proj{i}_{A}$ be a spectral decomposition of $\sigma_{A}$ where the eigenvalues $t_{i}$ are ordered non-increasingly. Define the projector $\Pi_{A}^{k}=\sum_{i\geq k}\proj{i}_{A}$, let $j$ be the smallest index such that $\mathrm{tr}\left[\Pi_{A}^{j}\sigma_{A}\right]\leq\eps$, and define $\Pi_{A}=\id_{A}-\Pi_{A}^{j}$ as well as $\bar{\sigma}_{A}=\Pi_{A}\sigma_{A}\Pi_{A}$. By~\cite[Lemma 13]{Berta09} we have
\begin{align}
H_{\max}(A)_{\sigma}>-\log\sup\{\lambda:\bar{\sigma}_{A}\geq\lambda\cdot\bar{\sigma}_{A}^{0}\}-2\cdot\log\frac{1}{\eps}&\geq\log\mathrm{tr}\left[\bar{\sigma}_{A}^{0}\right]-2\cdot\log\frac{1}{\eps}\\
&=H_{0}(A)_{\bar{\sigma}}-2\cdot\log\frac{1}{\eps}\ ,
\end{align}
and furthermore
\begin{align}
P(\bar{\sigma}_{A},\rho_{A})&\leq P(\sigma_{A},\rho_{A})+P(\bar{\sigma}_{A},\sigma_{A})\leq\eps'+P(\Pi_{A}\sigma_{A}\Pi_{A},\sigma_{A})\leq\eps'+\sqrt{1-\left(\mathrm{tr}\left[\Pi_{A}^{2}\sigma_{A}\right]\right)^{2}}\\
&\leq\eps'+\sqrt{1-\left(1-\eps\right)^{2}}\leq\eps'+\sqrt{2\eps}\ ,
\end{align}
where we used the triangle inequality for the purified distance, and a gentle measurement lemma for the purified distance (Lemma~\ref{lem:gentle}). Thus, we have
\begin{align}
H_{\max}^{\eps'}(A)_{\rho}=H_{\max}(A)_{\sigma}>H_{0}(A)_{\bar{\sigma}}-2\cdot\log\frac{1}{\eps}\geq H_{0}^{\eps'+\sqrt{2\eps}}(A)_{\rho}-2\cdot\log\frac{1}{\eps}\ .
\end{align}
\end{proof}

The zero-R\'enyi entropy can be smoothed by applying a projection.

\begin{lemma}\label{lem:simsmoothing}
Let $\eps\geq0$, and let $\rho_{A}\in\cS(A)$. Then there exists $\Pi_{A}\in\cP(A)$ with $\Pi_{A}\leq\id_{A}$, diagonal in any eigenbasis of $\rho_{A}$,
\begin{align}
H_{0}^{\eps}(A)_{\rho}\geq H_{0}(A)_{\Pi\rho\Pi}\ ,
\end{align}
and $\Pi_{A}\rho_{A}\Pi_{A}\in\cB^{\sqrt{4\eps}}(\rho_{A})$.
\end{lemma}

\begin{proof}
The idea for the proof is from the supplementary material~\cite[Lemma 14]{Berta09}. Let $\sigma_{A}\in\cB^{\eps}(\rho_{A})$ such that $H_{0}^{\eps}(A)_{\rho}=H_{0}(A)_{\sigma}$. It follows from the supplementary material~\cite[Lemma 8]{Berta09}, that $\sigma_{A}$  can be taken to be diagonal in any eigenbasis of $\rho_{A}$. Define
\begin{align}
\bar{\sigma}_{A}=\sigma_{A}-\{\sigma_{A}-\rho_{A}\}_{+}=\rho_{A}-\{\rho_{A}-\sigma_{A}\}_{+}\ ,
\end{align}
where $\{\cdot\}$ denotes the positive part of an operator. This implies $\bar{\sigma}_{A}\leq\sigma_{A}$, and we then have $H_{0}^{\eps}(A)_{\rho}\geq H_{0}(A)_{\bar{\sigma}}$. Since $\bar{\sigma}_{A}$ and $\rho_{A}$ also have the same eigenbasis, it follows that there exists $\Pi_{A}\in\cP(A)$ with $\Pi_{A}\leq\id_{A}$ such that $\bar{\sigma}_{A}=\Pi_{A}\rho_{A}\Pi_{A}$. Furthermore, we get by the equivalence of the trace distance and the purified distance (Lemma~\ref{lem:distance}) that
\begin{align}
P(\rho_{A},\bar{\sigma}_{A})&\leq\sqrt{\|\rho_{A}-\bar{\sigma}_{A}\|_{1}+\left|\mathrm{tr}\left[\rho_{A}\right]-\mathrm{tr}\left[\bar{\sigma}_{A}\right]\right|}=\sqrt{2\cdot\mathrm{tr}\left[\left\{\rho_{A}-\sigma_{A}\right\}_{+}\right]}\leq\sqrt{2\cdot\|\rho_{A}-\sigma_{A}\|_{1}}\\
&\leq\sqrt{4\cdot P(\rho_{A},\sigma_{A})}\leq\sqrt{4\eps}\ .
\end{align}
\end{proof}

The fully quantum asymptotic equipartition property for the smooth max-information and the smooth max-entropy is as follows.

\begin{lemma}\cite[Lemma B.21]{Berta09_2}\cite[Theorem 9]{Tomamichel08}\label{lem:aep}
Let $\eps>0$, $n\geq2\cdot(1-\eps^{2})$, and $\rho_{AB}\in\cS(AB)$. Then we have that
\begin{align}
&\frac{1}{n}I_{\max}^{\eps}(A:B)_{\rho^{\otimes n}}\leq I(A:B)_{\rho}+\frac{\xi(\eps)}{\sqrt{n}}-\frac{2}{n}\cdot\log\frac{\eps^{2}}{24}\\
&\frac{1}{n}H_{\max}^{\eps}(A)_{\rho^{\otimes n}}\leq H(A)_{\rho}+\frac{\eta(\eps)}{\sqrt{n}}\ ,
\end{align}
where $\xi(\eps)=8\sqrt{13-4\cdot\log\eps}\cdot(2+\frac{1}{2}\cdot\log|A|)$, and $\eta(\eps)=4\sqrt{1-2\cdot\log\eps}\cdot(2+\frac{1}{2}\cdot\log|A|)$.
\end{lemma}

%%%%%%%%%%%%%%%%%%%%%%%%%%%%%%%%%%%%%%%%%%%%%

\section{Misc Lemmas}

The following gives lower and upper bounds to the purified distance in terms of the trace distance.

\begin{lemma}\cite[Lemma 6]{Tomamichel09}\label{lem:distance}
Let $\rho,\sigma\in\cS_{\leq}(A)$. Then we have that
\begin{align}
\frac{1}{2}\cdot\|\rho_{A}-\sigma_{A}\|_{1}\leq P(\rho_{A},\sigma_{A})\leq\sqrt{\|\rho_{A}-\sigma_{A}\|_{1}+|\mathrm{tr}[\rho_{A}]-\mathrm{tr}[\sigma_{A}]|}\ .
\end{align}
\end{lemma}

The purified distance is convex in its arguments in the following sense.

\begin{lemma}\cite[Lemma A.3]{Berta09_2}\label{lem:conv}
Let $\rho^{i}_{A}$, $\sigma^{i}_{A}\in\cS_{\leq}(A)$ be with $\rho^{i}_{A}\approx_{\eps}\sigma^{i}_{A}$ for $i\in I$, and $\{p_{i}\}_{i\in I}$ a probability distribution. Then we have that
\begin{align}
\sum_{i\in I}p_{i}\rho^{i}_{A}\approx_{\eps}\sum_{i\in I}p_{i}\sigma^{i}_{A}\ .
\end{align}
\end{lemma}

The following is a gentle measurement lemma for the purified distance.

\begin{lemma}\cite[Lemma 7]{Berta09}\label{lem:gentle}
Let $\rho_{A}\in\cS(A)$, and $\Pi_{A}\in\cP(A)$ with $\Pi_{A}\leq\id_{A}$. Then we have that
\begin{align}
P(\rho_{A},\Pi_{A}\rho_{A}\Pi_{A})\leq\sqrt{1-\left(\mathrm{tr}\left[\Pi_{A}^{2}\rho_{A}\right]\right)^{2}}\ .
\end{align}
\end{lemma}

%%%%%%%%%%%%%%%%%%%%%%%%%%%%%%%%%%%%%%%%%%%%%

\section{Extractors Based on Permutations}

The following proposition concerns permutation-based extractors (operations that extract
uniform randomness independent of an adversary's information), and it is critical in establishing
our protocol for state merging of classically coherent states.

\begin{proposition}\cite[Section 5.2]{Szehr11}\label{prop:extractor}
Let $\rho_{XR}\in\cS(XR)$ be classical on $X$ with respect to $\{\ket{x}\}_{x\in X}$, and $X=X_{1}X_{2}$. Then we have that
\begin{align}
\frac{1}{|X|!}\cdot\sum_{P_{X}\in\mathbb{P}(X)}\left\|\mathrm{tr}_{X_{2}}\left[(P_{X}\otimes\id_{R})\rho_{XR}(P_{X}^{\dagger}\otimes\id_{R})\right]-\frac{\id_{X_{1}}}{|X_{1}|}\otimes\rho_{R}\right\|_{1}\leq\sqrt{|X_{1}|\cdot2^{-H_{\min}(X|R)_{\rho}}}\ ,
\end{align}
where $\mathbb{P}(X)$ denotes the group of permutations matrices on $\cH_{X}$ with respect to $\{\ket{x}\}_{x\in X}$, defined as $P(\pi)\ket{x}=\ket{\pi(x)}$ for $\pi\in S_{|X|}$, the symmetric group on $\{1,2,\ldots,|X|\}$.
\end{proposition}

%%%%%%%%%%%%%%%%%%%%%%%%%%%%%%%%%%%%%%%%%%%%%

\section{The Post-Selection Technique}

The following proposition lies at the heart of the post-selection technique for quantum channels.

\begin{proposition}\cite{ChristKoenRennerPostSelect}\label{prop:selection}
Let $\eps>0$,  and let $\cE^{n}_{A}$ and $\cF^{n}_{A}$ be quantum channels from $\cL(A^{\otimes n})$ to $\cL(B)$. If there exists a quantum channel $K_{\pi}$ for any permutation $\pi$ such that $(\cE^{n}_{A}-\cF^{n}_{A})\circ\pi=K_{\pi}\circ(\cE^{n}_{A}-\cF^{n}_{A})$, then $\cE^{n}_{A}$ and $\cF^{n}_{A}$ are $\eps$-close whenever
\begin{align}
\left\|((\cE^{n}_{A}-\cF^{n}_{A})\otimes\cI_{RR'})(\zeta^{n}_{ARR'})\right\|_{1}\leq\eps(n+1)^{-(|A|^{2}-1)}\ ,
\end{align}
where $\zeta^{n}_{ARR'}$ is a purification of the de Finetti state $\zeta_{AR}^{n}=\int\sigma_{AR}^{\otimes n}d(\sigma_{AR})$ with $\sigma_{AR}\in\cV(AR)$, $A\cong R$ and $d(\cdot)$ the measure on the normalized pure states on $AR$ induced by the Haar measure on the unitary group acting on $AR$, normalized to $\int d(\cdot)=1$. Furthermore, we can assume without loss of generality that $|R'|\leq(n+1)^{|A|^{2}-1}$.
\end{proposition}

A straightforward application of Carath\'eodory's theorem gives the following.

\begin{lemma}\cite[Corollary D.6]{Berta09_2}\label{lem:cara}
Let $\zeta_{AR}^{n}=\int\sigma_{AR}^{\otimes n}d(\sigma_{AR})$ as in Proposition~\ref{prop:selection}. Then we have that $\zeta_{AR}^{n}=\sum_{i}p_{i}\left(\omega^{i}_{AR}\right)^{\otimes n}$ with $\omega^{i}_{AR}\in\cV(AR)$, $i\in\{1,2,\ldots,(n+1)^{2|A||R|-2}\}$, and $\{p_{i}\}$ a probability distribution.
\end{lemma}

%%%%%%%%%%%%%%%%%%%%%%%%%%%%%%%%%%%%%%%%%%%%%

\section{Typical Projectors}\label{app:typical}

\label{sec:typ-review}A sequence $x^{n}$\ is typical with respect to some
probability distribution $p_{X}\left(  x\right)  $ if its empirical
distribution has maximum deviation $\delta$ from $p_{X}\left(  x\right)  $.
The typical set $T_{\delta}^{X^{n}}$ is the set of all such sequences:%
\begin{equation}
T_{\delta}^{X^{n}}\equiv\left\{  x^{n}:\left\vert \frac{1}{n}N\left(
x|x^{n}\right)  -p_{X}\left(  x\right)  \right\vert \leq\delta\ \ \ \ \forall
x\in\mathcal{X}\right\}  ,
\end{equation}
where $N\left(  x|x^{n}\right)  $ counts the number of occurrences of the
letter $x$ in the sequence $x^{n}$. The above notion of typicality is the
\textquotedblleft strong\textquotedblright\ notion (as opposed to the weaker
\textquotedblleft entropic\textquotedblright\ version of typicality sometimes
employed~\cite{CT91}). The typical set enjoys three useful properties:\ its
probability approaches unity in the large $n$ limit, it has exponentially
smaller cardinality than the set of all sequences, and every sequence in the
typical set has approximately uniform probability. That is, suppose that
$X^{n}$ is a random variable distributed according to $p_{X^{n}}\left(
x^{n}\right)  \equiv p_{X}\left(  x_{1}\right)  \cdots p_{X}\left(
x_{n}\right)  $, $\epsilon$ is positive number that becomes arbitrarily small
as $n$ becomes large, and $c$ is some positive constant. Then the following
three properties hold \cite{CT91}%
\begin{align}
\Pr\left\{  X^{n}\in T_{\delta}^{X^{n}}\right\}   &  \geq1-\epsilon
,\label{eq:typ-1}\\
\left\vert T_{\delta}^{X^{n}}\right\vert  &  \leq2^{n\left[  H\left(
X\right)  +c\delta\right]  },\label{eq:typ-2}\\
\forall x^{n}\in T_{\delta}^{X^{n}}:\ \ \ 2^{-n\left[  H\left(  X\right)
+c\delta\right]  }  &  \leq p_{X^{n}}\left(  x^{n}\right)  \leq2^{-n\left[
H\left(  X\right)  -c\delta\right]  }. \label{eq:typ-3}%
\end{align}

These properties translate straightforwardly to the quantum setting by
applying the spectral theorem to a density operator $\rho$. That is, suppose
that%
\begin{equation}
\rho\equiv\sum_{x}p_{X}\left(  x\right)  \left\vert x\right\rangle
\left\langle x\right\vert ,
\end{equation}
for some orthonormal basis $\left\{  \left\vert x\right\rangle \right\}  _{x}%
$. Then there is a typical subspace defined as follows:%
\begin{equation}
T_{\rho,\delta}^{n}\equiv\text{span}\left\{  \left\vert x^{n}\right\rangle
:\left\vert \frac{1}{n}N\left(  x|x^{n}\right)  -p_{X}\left(  x\right)
\right\vert \leq\delta\ \ \ \ \forall x\in\mathcal{X}\right\}  ,
\end{equation}
and let $\Pi_{\rho,\delta}^{n}$ denote the projector onto it. Then properties
analogous to (\ref{eq:typ-1}-\ref{eq:typ-3}) hold for the typical subspace.
The probability that a tensor power state $\rho^{\otimes n}$\ is in the
typical subspace approaches unity as $n$ becomes large, the rank of the
typical projector is exponentially smaller than the rank of the full $n$-fold
tensor-product Hilbert space of $\rho^{\otimes n}$, and the state
$\rho^{\otimes n}$ \textquotedblleft looks\textquotedblright\ approximately
maximally mixed on the typical subspace:%
\begin{align}
\text{Tr}\left\{  \Pi_{\rho,\delta}^{n}\ \rho^{\otimes n}\right\}   &
\geq1-\epsilon,\label{eq:typ-q-1}\\
\text{Tr}\left\{  \Pi_{\rho,\delta}^{n}\right\}   &  \leq2^{n\left[  H\left(
B\right)  +c\delta\right]  },\label{eq:typ-q-2}\\
2^{-n\left[  H\left(  B\right)  +c\delta\right]  }\ \Pi_{\rho,\delta}^{n}  &
\leq\Pi_{\rho,\delta}^{n}\ \rho^{\otimes n}\ \Pi_{\rho,\delta}^{n}%
\leq2^{-n\left[  H\left(  B\right)  -c\delta\right]  }\ \Pi_{\rho,\delta}^{n},
\label{eq:typ-q-3}%
\end{align}
where $H\left(  B\right)  $ is the entropy of $\rho$.

Suppose now that we have an ensemble of the form $\left\{  p_{X}\left(
x\right)  ,\rho_{x}\right\}  $, and suppose that we generate a typical
sequence $x^{n}$ according to a ``pruned'' distribution (defined as a normalized version
of $p_{X^n}(x^n)$ with support on its typical set and zero otherwise),
leading to a tensor product state $\rho_{x^{n}}\equiv\rho_{x_{1}}%
\otimes\cdots\otimes\rho_{x_{n}}$. Then there is a conditionally typical
subspace with a conditionally typical projector defined as follows:%
\begin{equation}
\Pi_{\rho_{x^{n}},\delta}^{n}\equiv\bigotimes\limits_{x\in\mathcal{X}}%
\Pi_{\rho_{x},\delta}^{I_{x}},
\end{equation}
where $I_{x}\equiv\left\{  i:x_{i}=x\right\}  $ is an indicator set that
selects the indices $i$\ in the sequence $x^{n}$ for which the $i^{\text{th}}$
symbol $x_{i}$\ is equal to $x\in\mathcal{X}$ and $\Pi_{\rho_{x},\delta
}^{I_{x}}$ is the typical projector for the state $\rho_{x}$. The
conditionally typical subspace has the three following properties:%
\begin{align}
\text{Tr}\left\{  \Pi_{\rho_{x^{n}},\delta}^{n}\ \rho_{x^{n}}\right\}   &
\geq1-\epsilon,\\
\text{Tr}\left\{  \Pi_{\rho_{x^{n}},\delta}^{n}\right\}   &  \leq2^{n\left[
H\left(  B|X\right)  +c\delta\right]  },\\
2^{-n\left[  H\left(  B|X\right)  +c\delta\right]  }\ \Pi_{\rho_{x^{n}}%
,\delta}^{n}  &  \leq\Pi_{\rho_{x^{n}},\delta}^{n}\ \rho_{x^{n}}\ \Pi
_{\rho_{x^{n}},\delta}^{n}\leq2^{-n\left[  H\left(  B|X\right)  -c\delta
\right]  }\ \Pi_{\rho_{x^{n}},\delta}^{n},
\end{align}
where $H\left(  B|X\right)  =\sum_{x}p_{X}\left(  x\right)  H\left(  \rho
_{x}\right)  $ is the conditional quantum entropy.

Let $\rho$ be the expected density operator of the ensemble $\left\{
p_{X}\left(  x\right)  ,\rho_{x}\right\}  $ so that $\rho=\sum_{x}p_{X}\left(
x\right)  \rho_{x}$. The following properties are proved in
Refs.~\cite{ieee2005dev,itit1999winter,W11}:%
\begin{align}
\forall x^{n} \in T^{X^{n}}_{\delta}: \text{Tr}\left\{  \rho_{x^{n}}%
\ \Pi_{\rho}\right\}   &  \geq1-\epsilon,\nonumber\\
\sum_{x^{n}}p_{X^{\prime n}}^{\prime}\left(  x\right)  \rho_{x^{n}}  &
\leq\left[  1-\epsilon\right]  ^{-1}\rho^{\otimes n}.
\label{eq:prune-avg-op-ineq}%
\end{align}

In order to justify some of the estimates made in
Section~\ref{sec:non-feedback-simulation},
we use the above estimates on eigenvalues and support
sizes. For the classical communication cost, we consider%
\begin{equation}
H_{R}\left(  W^{n}\right)  _{\bar{\gamma}^{i}}-H_{\min}\left(  W^{n}%
R^{n}\right)  _{\bar{\gamma}^{i}}+H_{0}\left(  R^{n}\right)  _{\bar{\gamma
}^{i}}.
\end{equation}
The smallest nonzero eigenvalue of the reduced state on $W^{n}$ is larger than
$2^{-n\left[  H\left(  W\right)  +c\delta\right]  }$ due to the typical
projection on $W^{n}$. Thus, we have that%
\begin{equation}
H_{R}\left(  W^{n}\right)  \leq n\left[  H\left(  W\right)  +c\delta\right]  .
\end{equation}
The largest eigenvalue of $\bar{\gamma}_{WR}^{i,n}$ is bounded by%
\begin{equation}
2^{-n\left[  H\left(  W\right)  _{\bar{\gamma}^{i}}-c\delta\right]
}2^{-n\left[  H\left(  R|W\right)  _{\bar{\gamma}^{i}}-c\delta\right]  },
\end{equation}
due to the typical projection on $W^{n}$ and the conditionally
typical projection on $R^{n}$. So we
have that%
\begin{equation}
H_{\min}\left(  W^{n}R^{n}\right)  _{\bar{\gamma}^{i}}\geq n\left[  H\left(
WR\right)  _{\bar{\gamma}^{i}}+2c\delta\right]  .
\end{equation}
The size of the support of $R^{n}$ is bounded from above by%
\begin{equation}
2^{n\left[  H\left(  R\right)  _{\bar{\gamma}^{i}}+\delta\right]  },
\end{equation}
due to the outermost projection on $R^{n}$.\ Thus, we have that%
\begin{equation}
H_{0}\left(  R^{n}\right)  _{\bar{\gamma}^{i}}\leq n\left[  H\left(  R\right)
_{\bar{\gamma}^{i}}+2c\delta\right]  .
\end{equation}
The above development then gives the following bound:%
\begin{equation}
H_{R}\left(  W^{n}\right)  _{\bar{\gamma}^{i}}-H_{\min}\left(  W^{n}%
R^{n}\right)  _{\bar{\gamma}^{i}}+H_{0}\left(  R^{n}\right)  _{\bar{\gamma
}^{i}}\leq n\left[  I\left(  W;R\right)  _{\bar{\gamma}^{i}}+5c\delta\right]
.
\end{equation}

We have similar arguments for bounding the shared randomness cost:
\begin{equation}
H_{R}\left(  W^{n}\right)  _{\bar{\gamma}^{i}}-H_{\min}\left(  W^{n}X^{n}%
R^{n}\right)  _{\bar{\gamma}^{i}}+H_{0}\left(  R^{n}X^{n}\right)
_{\bar{\gamma}^{i}}.
\end{equation}
By the same argument as above, we have that%
\begin{equation}
H_{R}\left(  W^{n}\right)  _{\bar{\gamma}^{i}}\leq n\left[  H\left(  W\right)
_{\bar{\gamma}^{i}}+c\delta\right]  .
\end{equation}
The largest eigenvalue of $\bar{\gamma}_{WXR}^{i,n}$ is bounded by%
\begin{align}
&  2^{-n\left[  H\left(  W\right)  _{\bar{\gamma}^{i}}-c\delta\right]
}2^{-n\left[  H\left(  X|W\right)  _{\bar{\gamma}^{i}}-c\delta\right]
}2^{-n\left[  H\left(  R|W\right)  _{\bar{\gamma}^{i}}-c\delta\right]  }\\
&  =2^{-n\left[  H\left(  WX\right)  _{\bar{\gamma}^{i}}-2c\delta\right]
}2^{-n\left[  H\left(  R|WX\right)  _{\bar{\gamma}^{i}}-c\delta\right]  }\\
&  =2^{-n\left[  H\left(  WXR\right)  _{\bar{\gamma}^{i}}-3c\delta\right]  },
\end{align}
where we have used the fact that $H\left(  R|W\right)  _{\bar{\gamma}^{i}%
}=H\left(  R|WX\right)  _{\bar{\gamma}^{i}}$ because the state on $R$ is
independent of $X$. Thus, we have that%
\begin{equation}
H_{\min}\left(  W^{n}X^{n}R^{n}\right)  _{\bar{\gamma}^{i}}\geq n\left[
H\left(  WXR\right)  _{\bar{\gamma}^{i}}-3c\delta\right]  .
\end{equation}
Finally, the support of $R^{n}X^{n}$ is bounded again by $2^{n\left[  H\left(
RX\right)  _{\bar{\gamma}^{i}}+2c\delta\right]  }$, due to the typical
projections, so that we have%
\begin{equation}
H_{0}\left(  R^{n}X^{n}\right)  _{\bar{\gamma}^{i}}\leq n\left[  H\left(
RX\right)  _{\bar{\gamma}^{i}}+2c\delta\right]  .
\end{equation}
The above development then gives the following bound:%
\begin{equation}
H_{R}\left(  W^{n}\right)  _{\bar{\gamma}^{i}}-H_{\min}\left(  W^{n}X^{n}%
R^{n}\right)  _{\bar{\gamma}^{i}}+H_{0}\left(  R^{n}X^{n}\right)
_{\bar{\gamma}^{i}}\leq n\left[  I\left(  W;XR\right)  _{\bar{\gamma}^{i}%
}+6c\delta\right]  .
\end{equation}

%%%%%%%%%%%%%%%%%%%%%%%%%%%%%%%%%%%%%%%%%%%%

\section{Uncertainty Relation}
\begin{lemma}
  \label{lem:minmax}
  For every $\proj{\psi}_{ABR}\in\cV(ABR)$ and observable (measurement) ${Z}_A$, we have that
  \begin{align}
    &H_{\rm min}(A|B)_\psi+H_{\rm max}({Z}_A|R)_\psi\leq \log|A|,\quad
   \end{align}
\end{lemma}
\begin{proof}
  Define $\lambda=H_{\rm min}(A|B)_\psi$ and let  $\sigma_B\in\cS(B)$ be such that 
\begin{align}
 \psi_{AB}\leq 2^{-\lambda}\id_A\otimes \sigma_B.
\end{align}
The measurement procedure can be described by an isometry $U_{A\rightarrow ZA}$ whose action is specified by $U_{A\rightarrow ZA}\ket{z}_A=\ket{z}_Z\ket{z}_A$, where $\{\ket{z}\}$ are the basis states associated with the (projective) measurement. Applied to $\psi_{AB}$ this yields
\begin{align}
 \xi_{ZAB}&= U_{A\rightarrow ZA}\psi_{AB}U_{A\rightarrow ZA}^\dagger\\
&\leq 2^{-\lambda}\,U_{A\rightarrow ZA}(\id_A\otimes \sigma_B) U_{A\rightarrow ZA}^\dagger\\
&=2^{-\lambda}\sum_z \ket{z}\bra{z}_Z\otimes \ket{z}\bra{{z}}_A\otimes \sigma_B\\
& \leq 2^{-\lambda}\id_{ZA}\otimes \sigma_B\\
&= 2^{-(\lambda-\log|A|)}\id_Z\otimes \pi_A\otimes \sigma_B,
\end{align}
where $\pi_A=\id_A/|A|$.  Thus, $\mu=\lambda-\log|A|$ and $\pi_A\otimes\sigma_B$ are feasible for $H_{\rm min}(Z|AB)_\xi$, meaning 
\begin{align}
H_{\rm min}(Z|AB)_\xi\geq \lambda-\log|A|\ .
\end{align}
Therefore the first claim follows, since $H_{\rm min}(Z|AB)_\xi=-H_{\rm max}(Z|R)_\xi=-H_{\rm max}(Z_A|R)_\psi$. 
\end{proof}

%%%%%%%%%%%%%%%%%%%%%%%%%%%%%%%%%%%%%%%%%%%%%

\bibliographystyle{plain}
\bibliography{library}

\end{document}